	\documentclass[
		paper=a4,11pt,american,pagesize,%
		abstract=true,headinclude=true,headlines=0,footinclude=true,footlines=-5,twoside=semi,%
		DIV=12,%
	]{scrartcl}
	\def\docclass{koma}
	\def\version{arxiv}

	\def\draftmode{false} 

\usepackage[utf8]{inputenc}

\makeatletter
\usepackage{xifthen}

\newcommand\iflipics[2]{\ifthenelse{\equal{\docclass}{lipics}}{#1}{#2}}
\newcommand\ifkoma[2]{\ifthenelse{\equal{\docclass}{koma}}{#1}{#2}}
\newcommand\ifieee[2]{\ifthenelse{\equal{\docclass}{ieee}}{#1}{#2}}
\newcommand\ifsiam[2]{\ifthenelse{\equal{\docclass}{siam}}{#1}{#2}}
\newcommand\ifsiamsingle[2]{\ifthenelse{\equal{\docclass}{siam-single}}{#1}{#2}}
\newcommand\ifmysiam[2]{\ifthenelse{\equal{\docclass}{my-siam}}{#1}{#2}}
\newcommand\ifacm[2]{\ifthenelse{\equal{\docclass}{acm}}{#1}{#2}}
\newcommand\ifdcc[2]{\ifthenelse{\equal{\docclass}{dcc}}{#1}{#2}}
\newcommand\ifspringerjournal[2]{\ifthenelse{\equal{\docclass}{springer-journal}}{#1}{#2}}
\newcommand\iflncs[2]{\ifthenelse{\equal{\docclass}{lncs}}{#1}{#2}}
\ifthenelse{ \equal{\docclass}{lipics} \OR \equal{\docclass}{koma} \OR \equal{\docclass}{ieee} \OR \equal{\docclass}{siam} \OR \equal{\docclass}{my-siam} \OR \equal{\docclass}{acm} \OR \equal{\docclass}{dcc} \OR \equal{\docclass}{springer-journal} \OR \equal{\docclass}{lncs} }{
	\PackageInfo{paper}{Building paper with docclass = \docclass} 
}{
	\PackageWarning{paper}{docclass = "\docclass", but must be one of "lipics", "koma", "ieee", "siam", "siam-single", "my-siam", "acm", "dcc", "springer-journal", "lncs"}
}

\newcommand\ifmanuscript[2]{\ifthenelse{\equal{\version}{manuscript}}{#1}{#2}}
\newcommand\ifarxiv[2]{\ifthenelse{\equal{\version}{arxiv}}{#1}{#2}}
\newcommand\ifsubmission[2]{\ifthenelse{\equal{\version}{submission}}{#1}{#2}}
\newcommand\ifproceedings[2]{\ifthenelse{\equal{\version}{proceedings}}{#1}{#2}}
\ifthenelse{ 
	\equal{\version}{manuscript} 
	\OR \equal{\version}{arxiv} 
	\OR \equal{\version}{submission} 
	\OR \equal{\version}{proceedings} 
}{
	\PackageInfo{paper}{Building paper version = \version} 
}{
	\PackageWarning{paper}{version = "\version", but must be one of "manuscript", "arxiv", "submission", "proceedings"}
}

\newcommand\ifdraft[2]{\ifthenelse{\equal{\draftmode}{true}}{#1}{#2}}
\ifthenelse{ \equal{\draftmode}{true} \OR \equal{\draftmode}{false} }{
	\PackageInfo{paper}{Building paper with draftmode = \draftmode} 
}{
	\PackageWarning{paper}{draftmode = "\draftmode", but must be "true" or "false"}
}

\usepackage[T1]{fontenc}
\ifsiam{
	\usepackage{lmodern}
	\usepackage{slantsc}
}{}
\ifsiamsingle{
	\usepackage{lmodern}
	\usepackage{slantsc}
}{}
\ifmysiam{
	\usepackage{lmodern}
	\usepackage{slantsc}
}{}
\ifkoma{
	\usepackage{lmodern}
	\usepackage{slantsc}
}{}
\iflipics{
	\usepackage{lmodern}
	\usepackage{slantsc}
}{}
\ifdcc{
	\usepackage{lmodern}
	\usepackage{slantsc}
}{}
\ifspringerjournal{
	\usepackage{lmodern}
	\usepackage{slantsc}
	\usepackage{xcolor}
}{}
\iflncs{
	\usepackage{lmodern}
	\usepackage{slantsc}
	\usepackage{xcolor}
}{}

\usepackage{babel}
\input{ushyphex.tex} 

\usepackage{array,multicol,multirow}
\ifieee{
	\usepackage[cmex10]{amsmath,mathtools}
	\usepackage{amsfonts,amssymb}
}{}
\ifkoma{
	\usepackage{amsmath,amsfonts,amssymb,mathtools}
}{}
\iflipics{
	\usepackage{amsmath,amsfonts,amssymb,mathtools}
}{}
\ifsiam{
	\usepackage{amsmath,amsfonts,amssymb,mathtools}
}{}
\ifsiamsingle{
	\usepackage{amsmath,amsfonts,amssymb,mathtools}
}{}
\ifmysiam{
	\usepackage{amsmath,amsfonts,amssymb,mathtools}
}{}
\ifacm{
	\usepackage{mathtools}
}{}
\ifdcc{
	\usepackage{amsmath,amsfonts,amssymb,mathtools}
}{}
\ifspringerjournal{
	\usepackage{amsmath,amsfonts,amssymb,mathtools}
}{}
\iflncs{
	\usepackage{amsmath,amsfonts,amssymb,mathtools}
}{}

\usepackage{mleftright}\mleftright 
\usepackage{relsize,xspace,booktabs,adjustbox,needspace,pbox,relsize,makecell}
\ifieee{
		
		\usepackage{enumitem}
}{
	\ifspringerjournal{
		\usepackage{enumitem}
		\setlist{topsep=\medskipamount}
	}{
		\usepackage{enumitem}
	}
}
\usepackage{graphicx}

\ifacm{}{
	\usepackage{colonequals}
}

\usepackage{wref}

\ifsiam{
	\usepackage{ltexpprt}
}{}
\ifsiamsingle{
	\usepackage{ltexpprt}
}{}

\newdimen\makeboxdimen
\newcommand\makeboxlike[3][l]{%
\setbox0=\hbox{#2}%
\global\makeboxdimen=\wd0%
\setbox1=\hbox{\makebox[\makeboxdimen][#1]{%
\makebox[0pt][#1]{#3}%
}}%
\ht1=\ht0%
\dp1=\dp0%
\box1%
}

\ifthenelse{\equal{\docclass}{koma} \OR \equal{\docclass}{my-siam}}{
	\setlength\parindent{1.5em}
	\usepackage[headsepline]{scrlayer-scrpage}
	\pagestyle{scrheadings}
	\clearscrheadfoot
	\AtBeginDocument{%
		\automark[section]{}%
	}
	\ohead{\pagemark}
	\rehead{\mytitle}
	\lohead{\headmark}
	\addtokomafont{caption}{\sffamily\small}
	\addtokomafont{captionlabel}{\sffamily\textbf}
	\setcapmargin{2em}
}{}
\ifmysiam{
	\setcapmargin{1em}
	\setcapindent{0em}
}{}
\ifmysiam{
	\setlength\parskip{0pt}
	\RedeclareSectionCommand[
		beforeskip=-1.25\baselineskip,
		afterskip=0.75\baselineskip,
	]{section}
	\RedeclareSectionCommand[
		beforeskip=-1\baselineskip,
		afterskip=-1.5em,
	]{subsection}
	\RedeclareSectionCommand[
		beforeskip=-1\baselineskip,
		afterskip=-1.5em,
	]{subsubsection}
	\RedeclareSectionCommand[
		beforeskip=-.25\baselineskip,
		indent=1.5em,
		afterskip=-1em,
	]{paragraph}
}{}
\ifdcc{
	\ifproceedings{}{
		\pagestyle{plain}
		\setlength{\footskip}{5ex}
	}
}{}

\AtBeginDocument{%
	\let\mytitle\@title%
}

\ifsiam{
	\usepackage[bibtex-alpha]{url-doi-arxiv}
}{}
\ifsiamsingle{
	\usepackage[bibtex-alpha]{url-doi-arxiv}
}{}
\ifmysiam{
	\usepackage[bibtex-alpha]{url-doi-arxiv}
}{}
\ifkoma{
	\usepackage[bibtex-alpha]{url-doi-arxiv}
}{}
\iflipics{
	\usepackage[bibtex]{url-doi-arxiv}
}{}
\ifdcc{
	\usepackage[bibtex-alpha]{url-doi-arxiv}
}{}
\ifspringerjournal{
	\usepackage[bibtex-alpha]{url-doi-arxiv}
}{}
\iflncs{
	\usepackage[bibtex-alpha]{url-doi-arxiv}
}{}

\let\oldthebibliography\thebibliography
\renewcommand\thebibliography[1]{%
	\oldthebibliography{#1}%
	\pdfbookmark[1]{References}{}%
}

\usepackage{lscape} 

\ifkoma{
	\usepackage{float}
	\floatstyle{plain}
	\usepackage{newfloat}
	\DeclareFloatingEnvironment[%
			name=Algorithm,%
			placement=thb,%
		]{algorithm}
}{}
\iflipics{
	\usepackage{newfloat}
	\DeclareFloatingEnvironment[%
			name=Algorithm,%
			placement=thb,%
		]{algorithm}
}{}
\ifmysiam{
	\usepackage{newfloat}
	\DeclareFloatingEnvironment[%
			name=Algorithm,%
			placement=thb,%
		]{algorithm}
}

\ifdcc{
	\usepackage{float}
	\usepackage[font={small},labelfont=bf]{caption}
}{}
\ifmysiam{
	\setcounter{topnumber}{3}
	\setcounter{bottomnumber}{3}
	\setcounter{totalnumber}{3}     
	\setcounter{dbltopnumber}{3}    

}{}
\ifsiam{

	\setcounter{topnumber}{3}
	\setcounter{bottomnumber}{3}
	\setcounter{totalnumber}{3}     
	\setcounter{dbltopnumber}{3}    

}{}
\ifsiamsingle{

	\setcounter{topnumber}{3}
	\setcounter{bottomnumber}{3}
	\setcounter{totalnumber}{3}     

}{}

\usepackage{dcolumn}

\usepackage{textcomp} 
\usepackage{listings}

\usepackage{tikz}

\usetikzlibrary{positioning,arrows.meta,fit}
\usetikzlibrary{backgrounds,calc,trees,graphs}
\usetikzlibrary{shapes.geometric,shapes.misc}
\usetikzlibrary{decorations,decorations.pathreplacing}
\usetikzlibrary{scopes}

\pgfdeclarelayer{background}
\pgfsetlayers{background,main}

\usetikzlibrary{external}
\tikzsetexternalprefix{pics/externalized/}
\tikzset{
	external/system call={%
		lualatex \tikzexternalcheckshellescape -halt-on-error %
			-interaction=batchmode -jobname "\image" "\texsource"%
	},
}
\tikzset{external/export=false} 

\iflipics{
	\newtheorem{fact}[theorem]{Fact}
	\newcommand\thmemph[1]{\emph{#1}}
	\newenvironment{proofof}[1]{%
		\begin{proof}[{{Proof of #1{}}}]%
	}{%
		\end{proof}%
	}
}{}
\ifacm{
	\AtEndPreamble{
		\theoremstyle{acmdefinition}
		\newtheorem{remark}[theorem]{Remark}
		\newtheorem{fact}[theorem]{Fact}
	}
	\newcommand\thmemph[1]{\emph{#1}}
	\newenvironment{proofof}[1]{%
		\begin{proof}[{{Proof of #1{}}}]%
	}{%
		\end{proof}%
	}
}{}
\ifsiam{
	\newtheorem{remark}{Remark}
	\newenvironment{proofof}[1]{%
		\begin{proof}[{{#1{}}}]%
	}{%
		\end{proof}%
	}
}{}
\ifsiamsingle{
	\newtheorem{remark}{Remark}
	\newenvironment{proofof}[1]{%
		\begin{proof}[{{#1{}}}]%
	}{%
		\end{proof}%
	}
}{}
\iflncs{
	\spnewtheorem{fact}[theorem]{Fact}{\itshape}{}
	\newcommand\thmemph[1]{\emph{#1}}
	\let\orig@endproof\endproof
	\def\endproof{\qed\orig@endproof}
	\newenvironment{proofof}[1]{%
		\begin{proof}[{{#1{}}}]%
	}{%
		\end{proof}%
	}
}{}
\ifthenelse{%
		\equal{\docclass}{lipics} \OR \equal{\docclass}{siam} \OR 
		\equal{\docclass}{siam-single} \OR \equal{\docclass}{acm} \OR
		\equal{\docclass}{lncs}%
}{}{
	\usepackage[amsmath,hyperref,thmmarks]{ntheorem}
	
	\theoremheaderfont{\sffamily\upshape\bfseries}
	\theorembodyfont{\slshape}
	\theoremseparator{:}
	\newtheoremstyle{proofstyle}%
	  {\item[\theorem@headerfont\hskip\labelsep ##1\theorem@separator]}%
	  {\item[\theorem@headerfont\hskip\labelsep ##3\theorem@separator]}
	
	\theorempreskip{\topsep} 

	\theoremsymbol{\adjustbox{scale=.8}{$\triangleleft\mkern-1mu$}}
	
	\newtheorem{theorem}{Theorem}[section]
	
	\theoremstyle{plain}
	\theorempreskip{\topsep}

	\newtheorem{lemma}[theorem]{Lemma}
	
	\newtheorem{corollary}[theorem]{Corollary}
	\newtheorem{definition}[theorem]{Definition}
	
	\theoremstyle{plain}
	\theorembodyfont{\upshape}

	\theoremsymbol{\raisebox{-.25ex}{$\Box$}}
	\qedsymbol{\raisebox{-.25ex}{$\Box$}}
	
	\theoremstyle{proofstyle}
	\theoremheaderfont{\sffamily\slshape}
	\newtheorem{proof}{Proof}

	\newcommand\thmemph[1]{\textit{#1}}
}

\iflipics{
		\newenvironment{thmenumerate}[2][]{%
			\begin{enumerate}[
				label={\textsf{\textbf{\color{darkgray}{\makebox[\widthof{(a)}][c]{\textup{(\alph*)}}}}}},
				ref={\ref{#2}\kern.1em--\kern.1em(\alph*)},
				itemsep=0pt,
				topsep=.5ex,
				leftmargin=1.75em,
				#1
			]%
		}{%
			\end{enumerate}%
		}
}{
	\ifspringerjournal{
		\newenvironment{thmenumerate}[2][]{%
			\begin{enumerate}[
				label={\makebox[\widthof{(a)}][c]{\textup{(\alph*)}}},
				ref={\ref{#2}\kern.1em--\kern.1em(\alph*)},
				itemsep=0pt,
				topsep=\smallskipamount,
				leftmargin=1.75em,
				#1
			]%
		}{%
			\end{enumerate}%
		}
	}{
		
	}
}

\newcommand*\ie{\mbox{i.\hspace{.2ex}e.}}
\newcommand*\eg{\mbox{e.\hspace{.2ex}g.}}

\newcommand*\wrt{\mbox{w.\hspace{.2ex}r.\hspace{.2ex}t.}\xspace}
\newcommand*\aka{\mbox{a.\hspace{.2ex}k.\hspace{.2ex}a.}\xspace}

\newcommand\R{\mathbb R}

\newcommand\Oh{O}

\usepackage{fixmath}

\newcommand{\ESymbol}{\mathbb{E}}

\newcommand{\ProbSymbol}{\ensuremath{\mathbb{P}}}
\providecommand{\given}{}
\DeclarePairedDelimiterXPP\Prob[1]{\ProbSymbol}[]{}{%
	\renewcommand\given{\nonscript\:\delimsize\vert\nonscript\:\mathopen{}}%
	#1%
}
\DeclarePairedDelimiterXPP\E[1]{\ESymbol}[]{}{%
	\renewcommand\given{\nonscript\:\delimsize\vert\nonscript\:\mathopen{}}%
	#1%
}
\DeclarePairedDelimiterXPP\Eover[2]{\ESymbol_{#1}}[]{}{%
	\renewcommand\given{\nonscript\:\delimsize\vert\nonscript\:\mathopen{}}%
	#2%
}
\DeclarePairedDelimiterXPP\ProbIn[2]{\ProbSymbol_{#1}}[]{}{%
	\renewcommand\given{\nonscript\:\delimsize\vert\nonscript\:\mathopen{}}%
	#2%
}
\providecommand{\Prob}{} 
\providecommand{\ProbIn}{} 
\providecommand{\E}{} 
\providecommand{\Eover}{} 

\newcommand{\surroundedmath}[3]{
	\mathchoice{
		#1{#2{#3}#2}%
	}{
		#1{#3}%
	}{
		#1{#3}%
	}{
		#1{#3}%
	}%
}

\newcommand\wrel[1]{\surroundedmath{\mathrel}{\;}{#1}}
\newcommand\wwrel[1]{\surroundedmath{\mathrel}{\;\;}{#1}}
\newcommand\bin[1]{\surroundedmath{\mathbin}{\:}{#1}}
\newcommand\wbin[1]{\surroundedmath{\mathbin}{\;}{#1}}

\iflipics{}{
	\makeatletter
	\let\oldalign\align
	\let\endoldalign\endalign
	\renewenvironment{align}{%
		\begingroup%
		\let\oldhalign\halign
		\def\halign{%
			\let\oldbreak\\%
			\def\nonnumberbreak{\nonumber\oldbreak*}%
			\def\\{%
				\@ifstar{\nonnumberbreak}{\oldbreak}%
			}%
			\oldhalign%
		}
		\oldalign%
	}{%
		\endoldalign%
		\endgroup%
	}
}
\newcommand*\numberthis[1][]{\stepcounter{equation}\tag{\theequation}}

\allowdisplaybreaks[3]

\newcommand\splitaftercomma[1]{%
  \begingroup
  \begingroup\lccode`~=`, \lowercase{\endgroup
    \edef~{\mathchar\the\mathcode`, \penalty0 \noexpand\hspace{0pt plus .25em}}%
  }\mathcode`,="8000 #1%
  \endgroup
}

\def\mydots{\xleaders\hbox to.5em{\hfill.\hfill}\hfill}
\newlength\tmpLenNotations

\ifdraft{%
	\iflipics{
		\usepackage{lineno} 
	}{
		\usepackage[switch]{lineno} 
	}
	\linenumbers
	\overfullrule=6mm
	
	\usepackage[color,notref,notcite]{showkeys}
	\definecolor{refkey}{gray}{.99}
	\colorlet{labelkey}{green!60!black!60}
	
	\usepackage[inline,nolabel]{showlabels}

	\showlabels{cite}
	\showlabels{citealt}
	\showlabels{citealp}
	\showlabels{citet}
	\showlabels{Citet}
	\showlabels{citep}
	\showlabels{citeauthor}
	\showlabels{Citeauthor}
	\showlabels{citefullauthor}
	\showlabels{citeyear}
	\showlabels{citeyearpar}
	\showlabels{wref}
	\showlabels{wpref}
	\showlabels{wtpref}
	\showlabels{wildref}
	\showlabels{wildpageref}
	\showlabels{wildtpageref}
}{}

\iflipics{
	\ifmanuscript{\hideLIPIcs}{}
	\ifarxiv{\hideLIPIcs}{}
	\ifsubmission{}{\nolinenumbers}
}{}

\iflncs{
    \ifsubmission{
        \usepackage[switch]{lineno}
        \linenumbers
    }{}
}{}

\ifdraft{}{%
	\usepackage{microtype}
}

\hypersetup{
	final,
	unicode=true, 
	bookmarksnumbered=true,
	bookmarksdepth=2,
	bookmarksopen=true,
	breaklinks=true,
	hidelinks,
}

\newsavebox\tmpbox

\ifdcc{
	\renewcommand\paragraph{\@startsection{paragraph}{4}{\parindent}
	                                      {\smallskipamount}
	                                      {-1em}%
	                                      {\normalfont\normalsize\bfseries}}
}{}
\iflipics{
	\let\oldparagraph\paragraph
	\renewcommand\paragraph[1]{%
		\oldparagraph*{\textcolor{lipicsGray}{#1.}}
	}
}{
	\let\oldparagraph\paragraph
	\renewcommand\paragraph[1]{%
		\oldparagraph{#1.}
	}
}

\ifmysiam{
	\let\oldsubsection\subsection
	\renewcommand\subsection[1]{%
		\oldsubsection{#1.}%
	}
	\let\oldsubsubsection\subsubsection
	\renewcommand\subsubsection[1]{%
		\oldsubsubsection{#1.}%
	}
}{}
\ifsiam{
	\let\oldsubsection\subsection
	\renewcommand\subsection[1]{%
		\oldsubsection{#1.}%
	}
	\let\oldsubsubsection\subsubsection
	\renewcommand\subsubsection[1]{%
		\oldsubsubsection{#1.}%
	}
}{}
\ifsiamsingle{
	\let\oldsubsection\subsection
	\renewcommand\subsection[1]{%
		\oldsubsection{#1.}%
	}
	\let\oldsubsubsection\subsubsection
	\renewcommand\subsubsection[1]{%
		\oldsubsubsection{#1.}%
	}
}{}

\let\epsilon\varepsilon

\def\myacknowledgements{}
\ifkoma{
	
}{}
\ifieee{
	
}{}
\ifsiam{
	
}{}
\ifsiamsingle{
	
}{}
\ifmysiam{
	
}{}
\ifdcc{
	
}{}
\ifacm{
	
}{}
\ifspringerjournal{
	
}{}
\iflncs{
	
}{}

\setlist[description]{font=\boldmath}

\usepackage{optidef}
\usepackage[most]{tcolorbox}
\usepackage{thm-restate}
\newcommand*\NP{\ensuremath{\textsf{NP}}\xspace}
\renewcommand*\P{\ensuremath{\textsf{P}}\xspace}

\newcommand{\problemtitle}[1]{\gdef\@problemtitle{#1}}
\newcommand{\problemvalue}[1]{\gdef\@problemvalue{#1}}
\newcommand{\probleminput}[1]{\gdef\@probleminput{#1}}
\newcommand{\problemquestion}[1]{\gdef\@problemquestion{#1}}

\NewEnviron{problem}{
  \problemtitle{}\probleminput{}\problemquestion{}\problemvalue{}
  \BODY
  \par\addvspace{.5\baselineskip}
  \noindent
  \begin{tabularx}{\textwidth}{@{\hspace{\parindent}} l X c}
    \multicolumn{2}{@{\hspace{\parindent}}l}{\@problemtitle} \\
    \textbf{Instance:} & \@probleminput \\
    \textbf{Definitions:} & \@problemvalue
    \\
    \textbf{Question:} & \@problemquestion
  \end{tabularx}
  \par\addvspace{.5\baselineskip}
}

\newcommand{\probtitle}[1]{
\multicolumn{2}{@{\hspace{\parindent}}l}{#1} \\[5pt]
}
\newcommand{\probrow}[2]{
\textbf{#1:} & {#2} \\[5pt]
}

\NewEnviron{prob}{
  \par\addvspace{.5\baselineskip}
  \noindent
  \begin{tabularx}{\textwidth}{@{\hspace{\parindent}} l X c}
  \BODY
  \end{tabularx}
  \par\addvspace{.5\baselineskip}
}

\tcolorboxenvironment{prob}{
  colback=white,
  colframe=black,
  boxrule=0.8pt,
  sharp corners,
  before skip=10pt,
  after skip=10pt
}

\tcolorboxenvironment{problem}{
  colback=white,
  colframe=black,
  boxrule=0.8pt,
  sharp corners,
  before skip=10pt,
  after skip=10pt
}

\DeclareMathOperator{\optlin}{\text{\textup{\textsc{Opt-Lin}}}}
\DeclareMathOperator{\optent}{\text{\textup{\textsc{Opt-Ent}}}}

\newcommand*\DegreeEntropy{\ensuremath{\mathcal H^{\text{in}}_{\deg}}\xspace}

\newcommand*\TheOptProblem{\textsc{MINETREX}\xspace} 
\newcommand*\TheOptProblemLongForm{\textsc{Minimum-Entropy Tree-Extraction}\xspace}

\newcommand*\rankop{\mathsf{rank}}
\newcommand*\selop{\mathsf{select}}
\newcommand*\accessop{\mathsf{access}}

\makeatother

\tikzstyle{graph}=[every node/.style={draw,circle,minimum size=5pt,inner sep=0pt,fill=black}]
\tikzstyle{digraph}=[graph, every path/.style={-stealth}]

\tikzstyle{chosen edge}=[every path/.append style={thick, densely dotted, red}]

\tikzstyle{macro edge}=[color=gray, double distance=2pt, arrows={-Latex[width'=0pt 0.7, length=9pt]}]

\iflipics{
}{
	\colorlet{lipicsGray}{gray}
}
\usepackage{tabularx}

\usepackage{hyperref}

\ifacm{
		\title{Rooting Out Entropy: Optimal Tree Extraction for Ultra-Succinct Graphs}
}{}
\ifspringerjournal{
	\title{Rooting Out Entropy: Optimal Tree Extraction for Ultra-Succinct Graphs}
}{
	\title{Rooting Out Entropy: Optimal Tree Extraction for Ultra-Succinct Graphs}
}

\ifthenelse{\NOT \equal{\docclass}{lipics} \AND \NOT \equal{\docclass}{acm} \AND \NOT \equal{\docclass}{springer-journal} \AND \NOT \equal{\docclass}{lncs} }{
	\newcommand\email[1]{\texttt{#1}}
	\author{%
		Ziad Ismaili Alaoui%
        \footnote{University of Liverpool, United Kingdom, 
		\email{ziad.ismaili-alaoui\,@\,liverpool.ac.uk}}
	\and
		Tamio-Vesa Nakajima%
			\footnote{University of Marburg, Germany, 
			\email{nakajima\,@\,informatik.uni-marburg.de}}
	\and
		Namrata%
            \footnote{University of Birmingham, United Kingdom, 
			\email{n.namrata\,@\,bham.ac.uk}}
	\and
		Sebastian Wild%
			\footnote{University of Marburg, Germany, 
			\email{wild\,@\,informatik.uni-marburg.de}}
	}
	
	\date{\small\today}
}{}

\begin{document}

\ifacm{}{\maketitle} 

\begin{abstract}
We combine two methods for the lossless compression of unlabeled graphs~-- entropy compressing adjacency lists and computing canonical names for vertices~-- and solve an ensuing novel optimisation problem:
\TheOptProblemLongForm (\TheOptProblem).
\TheOptProblem asks to determine a spanning forest $F$ to remove from a graph $G$ so that the remaining graph $G-F$ has minimal indegree entropy $H(d_1,\ldots,d_n) = \sum_{v\in V} d_v \log_2(m/d_v)$ among all choices for $F$.
(Here $d_v$ is the indegree of vertex $v$ in $G-F$; $m$ is the number of edges.)

We show that \TheOptProblem is \NP-hard to approximate with additive error better than $\delta n$ (for some constant $\delta>0$),
and provide a simple greedy algorithm that achieves additive error at most $n / \ln 2$.
By storing the extracted spanning forest and the remaining edges separately, 
we obtain a degree-entropy compressed (``ultrasuccinct'') data structure for representing an arbitrary (static) unlabeled graph
that supports navigational graph queries in logarithmic time.
It serves as a drop-in replacement for adjacency-list representations 
using substantially less space for most graphs; 
we precisely quantify these savings in terms of the maximal subgraph density.

Our inapproximability result uses an approximate variant of the hitting set problem on \emph{biregular} instances whose hardness proof is contained implicitly in a reduction by Guruswami and Trevisan (APPROX/RANDOM 2005);
we consider the unearthing of this reduction partner of independent interest with further likely uses in hardness of approximation.
\end{abstract}

\ifacm{%
	\maketitle%
}{}

\section{Introduction}
\label{sec:introduction}

In many applications of network analysis, only the \emph{structure} of the underlying graph is of importance~-- the names or labels of most vertices are often irrelevant.
With growing data sizes, it is clearly desirable not to waste space on storing such irrelevant names;
on ``edge devices'' it may even become imperative to work on compressed representations of data.
The thriving area of space-efficient data structures~--
and specifically succinct graph data structures~-- is testimony to this desire.

However, when ``writing down'' a general graph (not known to be from a restricted class of graphs), it seems inevitable that we use \emph{some} vertex identifiers by which we can refer to a vertex when describing the edges of the graph.
And even if we obtain these vertex ids as some canonical names derived solely from the graph's topology,
classical graph representations such as adjacency lists \emph{already} encode vertex ids implicitly (the \emph{order} of vertices in the sequential representation).
So is it even thinkable to save the space for useless labels \emph{without further assumptions on the stored graphs?}

We give a resounding \emph{yes} to this question: 
Our \emph{``tree-extraction''} data structure
is (1) simple enough to have practical appeal, (2) often (under conditions we analyze below) achieves an asymptotically optimal space reduction at (3) the expense of a very modest slowdown for operations and construction time.
\emph{Tree Extraction (TREX)} (\wref{fig:goodGraph}) combines two methods for the lossless compression of graphs~-- entropy compressing adjacency lists and computing canonical names for vertices.
This combination gives rise to a novel optimisation problem, whose study and exploitation is the main technical contribution of this paper:
\TheOptProblemLongForm (\TheOptProblem).

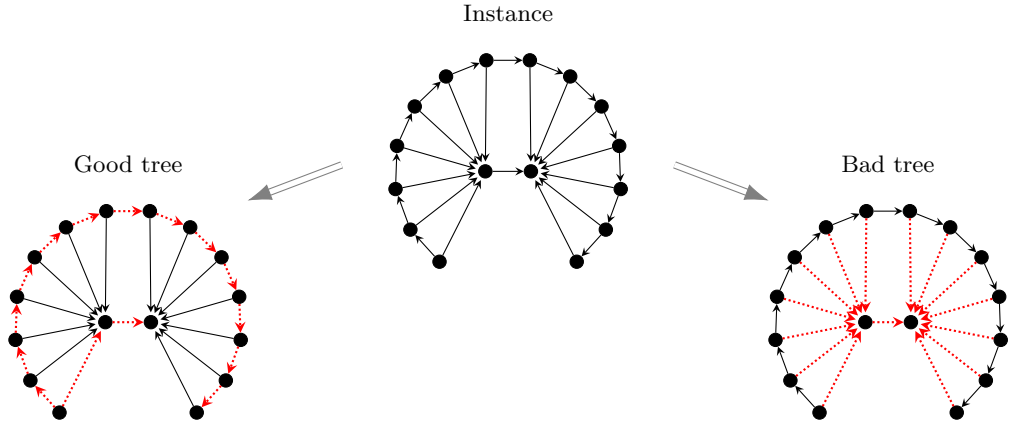
\begin{figure}[tb]
    \centering
    \begin{tikzpicture}
        \newcommand{\DrawPoints}{
            \node (A) at (-.5, 0){};
            \node (B) at (.5,0){};
            \node (X1) at (233:2.5){};
            \node (X2) at (211:2.5){};
            \node (X3) at (189:2.5 ){};
            \node (X4) at (167:2.5){};
            \node (X5) at (145:2.5){};
            \node (X6) at (123:2.5){};
            \node (X7) at (101:2.5){};
            \node (Y7) at (79:2.5){};
            \node (Y6) at (57:2.5){};
            \node (Y5) at (35:2.5){};
            \node (Y4) at (13:2.5){};
            \node (Y3) at (-9:2.5){};
            \node (Y2) at (-31:2.5){};
            \node (Y1) at (-53:2.5){};
        }
        \newcommand{\DrawOuterShell}{
            \draw (X1) -> (X2);
            \draw (X2) -> (X3);
            \draw (X3) -> (X4);
            \draw (X4) -> (X5);
            \draw (X5) -> (X6);
            \draw (X6) -> (X7);
            \draw (X7) -> (Y7);
            \draw (Y2) -> (Y1);
            \draw (Y3) -> (Y2);
            \draw (Y4) -> (Y3);
            \draw (Y5) -> (Y4);
            \draw (Y6) -> (Y5);
            \draw (Y7) -> (Y6);
        }
    
        \newcommand\DrawRibs{
        \begin{scope}[shorten >=1pt]
            \draw (X2) -> (A);
            \draw (X3) -> (A);
            \draw (X4) -> (A);
            \draw (X5) -> (A);
            \draw (X6) -> (A);
            \draw (X7) -> (A);
            \draw (Y1) -> (B);
            \draw (Y2) -> (B);
            \draw (Y3) -> (B);
            \draw (Y4) -> (B);
            \draw (Y5) -> (B);
            \draw (Y6) -> (B);
            \draw (Y7) -> (B);
        \end{scope}
        }
    
        \scoped[scale=.6, digraph, local bounding box=inst]{
            \DrawPoints{}
            \DrawOuterShell{}
            \DrawRibs{}
            \draw[shorten >=1pt] (X1) -> (A);
            \draw (A) -> (B);
        }
        
        \scoped[xshift=-5cm, yshift=-2cm, scale=.6, digraph, local bounding box=goodSol]{
            \DrawPoints{}
            \DrawRibs{}
            \scoped[chosen edge]{
                \DrawOuterShell{}
                \draw[shorten >=1pt] (X1) -> (A);
                \draw (A) -> (B);
            }
        }
            
        \scoped[xshift=5cm, yshift=-2cm, scale=.6, digraph, local bounding box=badSol]{
            \DrawPoints{}
            \DrawOuterShell{}
            \scoped[chosen edge]{
                \DrawRibs{}
                \draw[shorten >=1pt] (X1) -> (A);
                \draw (A) -> (B);
            }
        }
        
    \node[above=0.3cm of inst.north] {\footnotesize Instance};
    \node[above=0.3cm of goodSol.north] {\footnotesize Good tree};
    \node[above=0.3cm of badSol.north] {\footnotesize Bad tree};
    
    \node[left=.35cm of inst.west] (Q){};
    \node[right=0cm of badSol.north west] (W){};
    \node[right=.35cm of inst.east] (E){};
    \node[left=0cm of goodSol.north east] (R){};

    \draw[macro edge] (E) -> (W);
    \draw[macro edge] (Q) -> (R);
        
    \end{tikzpicture}
    \medskip
	\caption{%
		Example showing the significance of the tree in \TheOptProblemLongForm.
		In the example graph (\textbf{middle}), tree extraction can, depending on the chosen tree (red dotted), either (\textbf{left}) yield the optimal space saving of $\sim n \lg n$ bits for the entropy-compressed representation of the remaining edges (black), or (\textbf{right}) negligible saving of $O(n)$ bits, which is comparable to its overhead to store the extracted tree.
		We point out that the directions of edges in the extracted tree can be arbitrary; tree extraction is not limited to deleting an arborescence.
	}
    \label{fig:goodGraph}
\end{figure}

\paragraph{Labelled Graph Representations}

For ease of presentation, let $G=(V,E)$ be a \emph{directed} and connected%
\footnote{
	Connectivity simplifies description and analysis of our data structure; our ideas can easily be adapted for disconnected graphs, too.
} 
graph; we generalise our results to undirected graphs later.
We assume $V=[n] = \{1,\ldots,n\}$.
For $v\in V$, denote by $d_v = d^{\text{in}}_v$ the \emph{indegree} of $v$, \ie, the number of incoming edges $(w,v)\in E$.
A textbook adjacency-list representation of $G$ using linked lists uses $\Theta(m \log n)$ bits of space when $G$ has $n=|V|$ vertices and $m =|E| \ge n-1$ edges.
A more careful array-based representation using $m \lceil\lg n\rceil + n \lceil\lg m\rceil$ bits of space%
\footnote{%
	Here and throughout, we write $\lg$ for $\log_2$ and assume the word-RAM model with word size $\Theta(\log n)$.
} 
is easily possible: store all adjacency lists concatenated in one long array $A[1..m]$, plus an array $S[1..n]$ for the starting indices so that $N^+(v)$, the outneighbourhood of $v$ is found in $A[S[v]..S[v]+d^{\text{out}}_v)$.
By replacing the array $A$ of integers in $[1..n] = V$ with a 
\emph{wavelet tree}~\cite{GrossiGuptaVitter2003} (see also \wref{lem:wavelet-tree-simple}),
we obtain an entropy-compressed representation of $G$~\cite{Navarro2014}:
The total space becomes $m \mathcal \DegreeEntropy(G) + o(m) + n \lceil\lg m\rceil$ bits,
where 
\[
		\DegreeEntropy(G) 
	\wwrel= 
		\sum_{v=1}^n   \frac{d^{\text{in}}_v}{m} \lg \left( \frac{m}{d^{\text{in}}_v} \right)
\]
denotes the (zero-th order) empirical \emph{indegree entropy} of $G$.%
\footnote{%
	Here and throughout, we take the standard convention that $x \lg x = x \lg (1 / x) = 0$ for $x = 0$.%
}
As a bonus, with the wavelet tree's rank/select operations, we can navigating over \emph{incoming} edges~\cite{Navarro2014}; all operations take $O(\lg n)$ time.

Since \DegreeEntropy coincides with the empirical entropy of the \emph{adjacency string} $A[1]\ldots A[m]$ of edge targets,
it is a natural instance-specific measure of compressibility of $G$.
It is also the information-theoretic instance-specific lower bound in natural models of random graphs:
$m\DegreeEntropy(G) \sim \lg(1/\Prob{G})$ for $\Prob{G}$ the probability of $G$ to arise according to the \emph{directed configuration model} (fixed degree sequence model) and the (labelled) Barabási-Albert model of preferential attachment graphs.%
\footnote{%
	Here and throughout, $f(n)\sim g(n)$ iff $\lim f(n) / g(n) = 1$.
}
In both cases, \emph{for the labelled graph with given vertex ids}, a simple wavelet-tree based representation uses \emph{instance-optimal} space.

\paragraph{Tree Extraction}

We explore how to extend the above results to \emph{unlabelled} graphs.
The key idea of \emph{Tree Extraction (TREX)} is to store some of the graph edges as part of a (rooted, ordered) \emph{tree}~$T$, whose topology can be stored with 3 bits per edge (2 for the tree shape, 1 for the edge orientation)~-- much fewer than, say, the $\lg n$ bits in the adjacency string $A$.
We emphasise that we select an undirected spanning tree, not an arborescence; the edges need not be oriented towards the root. In other words, we select a subset of edges which, when ignoring edge orientations, form a spanning tree. 
Storing the shape of $T$ using $\sim2n$ bits of space can be achieved using any of the many succinct tree data structures, \eg, \cite{Jacobson1989,MunroRaman2001,GearyRamanRaman2006,NavarroSadakane2014,BenoitDemaineMunroRamanRamanRao2005,JanssonSadakaneSung2012,FarzanMunro2014,HeMunroNekrichWildWu2020}.

Normally, such savings are entirely moot, since we have to also store which tree node corresponds to which graph vertex~-- adding back $\lg(n!)\sim n \lg n$ bits of space.
But when we are \emph{not interested} in the original vertex ids, we might as well let $T$ choose ids of its personal convenience~-- \emph{canonical names} induced by the tree data structure~-- and use these names for storing the $m-n+1$ remaining edges in $E(G)\setminus E(T)$.
For the simple array-based adjacency lists with $\lceil\lg n\rceil$ bits per edge, tree extraction \emph{always} yields space savings of $\sim n \lg n$ bits irrespective of how we choose $T$,%
\footnote{%
	This explains why the (to our knowledge) first (implicit) use of tree extraction~-- in the GLOUDS data structure~-- simply hard-coded the use of a BFS tree~\cite{FischerPeters2016}.
}
but for the entropy-compressed wavelet-tree representation, the tree can make a huge difference (\wref{fig:goodGraph}). \emph{How should we choose the tree to extract? Can we obtain an instance-optimal unlabelled graph representation?} 

\paragraph{\TheOptProblemLongForm}

This first question is precisely the \TheOptProblemLongForm (\TheOptProblem) problem: given a directed, connected graph $G$, find a spanning tree $T$ that minimises $\DegreeEntropy(G-T)$.
(In principle, we could use any spanning \emph{forest}, but there always exists a spanning tree that is an optimal solution, cf.~\wref{app:tree-not-forest}.)

To our knowledge, this combinatorial optimisation problem has never been studied.
We show that \TheOptProblem is \NP-hard, and indeed, that there exists a constant $\delta>0$ such that it is \NP-hard to even find a solution of cost $\le \mathit{OPT}+\delta n$.
The former is by a simple reduction from \textsc{Exact Cover by 3-Regular 3-Sets (X3C)}; the latter extends this reduction to work with a biregular approximate variant of \textsc{Exact Hitting Set} (EHS) ($(\beta, r, k)$-\textsc{Almost}-EHS). 
Even though the hardness of this problem is implicitly contained in a known result~\cite{GuruswamiTrevisan2005}, 
to our knowledge, the explicit formulation of this latter gap problem is novel, 
which we consider of potential interest for other inapproximability reductions,
\eg{} turning X3C reductions into hardness-of-approximation reductions.

On the positive side, we show that a natural greedy algorithm works exceedingly well.
Any edge included in the tree does not have to be stored in the wavelet tree.
Ignoring the effect that its very removal has on the resulting empirical entropy, extracting an edge $(u,v)$
saves us $\lg(m/d_v) = \lg m - \lg d_v$ bits.
A greedy approach would thus sort the edges by decreasing savings $\lg(m/d_v)$~-- or equivalently, by increasing target degree $d_v$~-- and select an optimal spanning tree using these edge weights.
Despite being both myopic and ignorant to its own impact, 
we prove that this greedy MST approach
finds a spanning tree $T$ of \TheOptProblem cost $\le \mathit{OPT} + n / \ln 2$, and is thus, \wrt its approximation guarantee, optimal up to constants (among all polytime algorithms, unless $\P=\NP$).

\paragraph{TREX and Twin Removal Data Structure}

\begin{figure}[tb]
    \centering
    \usetikzlibrary{graphs.standard}
\usetikzlibrary{decorations.pathreplacing}

\colorlet{oldname}{lipicsGray}
\begin{tikzpicture}[
    	node distance = 0pt,
    	tree node/.style={circle,draw,inner sep=0pt,minimum size=14pt},
        graph edge/.style={very thick,->,>=Stealth},
        tree up edge/.style={graph edge,red,dashed},
        tree down edge/.style={graph edge,red,dashed},
    	scale=1.2,
    ]
    \node[shape=rectangle, rounded corners=2mm,draw=black] (1) at (-1,0) {2\textcolor{oldname}{:1}};
    \node[shape=rectangle, rounded corners=2mm,draw=black] (2) at (1,-0.3) {5\textcolor{oldname}{:2}};
    \node[shape=rectangle, rounded corners=2mm,draw=black] (3) at (4,-1) {6\textcolor{oldname}{:3}};
    \node[shape=rectangle, rounded corners=2mm,draw=black] (4) at (3.5,0.5) {3\textcolor{oldname}{:4}};
    \node[shape=rectangle, rounded corners=2mm,draw=black] (5) at (-2.8,-1.6) {7\textcolor{oldname}{:5}};
    \node[shape=rectangle, rounded corners=2mm,draw=black] (6) at (2,2) {1\textcolor{oldname}{:6}};
    \node[shape=rectangle, rounded corners=2mm,draw=black] (7) at (-1,-2) {8\textcolor{oldname}{:7}};
    \node[shape=rectangle, rounded corners=2mm,draw=black] (8) at (-2,2) {4\textcolor{oldname}{:8}};
    \draw[graph edge] 
            (1) edge (7) 
            (5) edge (7) 
            (8) edge[tree up edge] node[left,above, sloped,scale=.7]{$D[4]=1$} (1) 
            (5) edge (1) 
            (1) edge (2) 
            (4) edge[tree down edge] node[left,above, sloped,scale=.7]{$D[5]=0$} (2) 
            (3) edge (2) 
            (6) edge[tree down edge] node[left,above, sloped,scale=.7]{$D[3]=0$} (4) 
            (7) edge (2) 
            (5) edge[tree up edge] node[left,above, sloped,scale=.7]{$D[7]=1$} (8) 
            (6) edge (2) 
            (1) edge[tree up edge] node[left,above, sloped,scale=.7]{$D[2]=1$} (6) 
            (3) edge[tree down edge] node[left,above, sloped,scale=.7]{$D[8]=0$} (7) 
            (4) edge[tree down edge] node[left,above, sloped,scale=.7]{$D[6]=0$} (3)
    ;

    \node[] at (-4.2, 0) {$G =$};
    
    \node[anchor=base] at (-4, -3) {\llap{$A'$}${}={}$};
    \foreach [count=\i] \t in {5,5,8,5,2,8,5} {
    	\node[inner sep=3pt,anchor=base,minimum width=1.3em] (ap\i) at (-3.9+\i/2,-3) {$\t$} ;
    	\node[scale=.4] at ($(ap\i)+(0,.3)$) {\i};
    }
    \node[anchor=base] at (-4, -3.75) {\llap{$S'$}${}={}$};%
    \foreach [count=\i]\c in {1,0,1,0,0,1,1,1,1,0,1,0,0,1,0} {
    	\node[inner sep=1pt,anchor=base] (sp\i) at (-3.75+\i/4.3,-3.75) {\ttfamily\c} ;
    	\node[scale=.4] at ($(sp\i)-(0,.2)$) {\i};
    }
    \node[anchor=base] at (-4, -5.5) {\llap{$D$}${}={}$};
    \foreach [count=\i] \c in {{\color{gray}--},0,1,0,1,1,0,1} {
    	\node[inner sep=1pt,anchor=base] (d\i) at (-3.7-1/6+\i/3,-5.5) {\ttfamily\makeboxlike{0}{\c}} ;
    	\node[scale=.4] at ($(d\i)-(0,.2)$) {\i};
    }
	
	\foreach [count=\t] \f in {2,4,5,10,12,13,15} {
		\draw[densely dotted] (sp\f) to[out=90,in=-90,looseness=1.25] (ap\t);
	}
	\foreach \tofit in {{(ap1)},{(ap2) (ap3)},{(ap4)},{(ap5) (ap6)},{(ap7)}} {
		\node[black!50,draw,densely dotted,rectangle,rounded corners=2pt,fit={\tofit},inner sep=0pt] {};
	}

    \node[] at (1.1, -4) {$T=$};

    \node[tree node] (t1) at (4,-2.5) {1};
    \node[tree node] (t2) at (3.5,-3.4) {2};
    \node[tree node] (t3) at (4.5,-3.4) {3};
    \node[tree node] (t4) at (3,-4.3) {4};
    \node[tree node] (t5) at (4,-4.3) {5};
    \node[tree node] (t6) at (5,-4.3) {6};
    \node[tree node] (t7) at (2.5,-5.2) {7};
    \node[tree node] (t8) at (5.5,-5.2) {8};

    \draw[thick] (t1) edge (t2) (t1) edge (t3) (t2) edge (t4) (t4) edge (t7) (t3) edge (t5) (t3) edge (t6) (t6) edge (t8);

    \foreach [count=\v from 2] \p/\d in {%
                1/false,%
                1/true,%
                2/false,%
                3/true,%
                3/true,%
                4/false,%
                6/true%
    } {
        \ifthenelse{\equal{\d}{true}}{
            \node[tree node,scale=.7] (v\v) at ($(d\v)+(0,.4)$) {\v};
            \node[tree node,scale=.7] (pv\v) at ($(v\v)+(0,.6)$) {\p};
            \draw[tree down edge,thick] (pv\v) -- (v\v) ;
            \draw[tree down edge,thin,black,-{Stealth[length=8pt,red,sep=6pt]}] (t\p) -- (t\v);
        }{
            \node[tree node,scale=.7] (v\v) at ($(d\v)+(0,.4)$) {\v};
            \node[tree node,scale=.7] (pv\v) at ($(v\v)+(0,.6)$) {\p};
            \draw[tree up edge,thick] (v\v) -- (pv\v) ;
            \draw[tree up edge,thin,black,-{Stealth[length=8pt,red,sep=6pt]}] (t\v) -- (t\p);
        }
    }

    \draw[gray, densely dotted] (6) edge[bend left=90] (t1);
    \draw[gray, densely dotted] (1) edge[bend right=20] (t2);

    \draw [yshift=-.5pt,very thick, decorate,decoration={brace,amplitude=5pt,mirror,raise=3.3ex}]
    	(2.2,-2.2) -- (2.2,-5.8) node[midway,yshift=-3em]{};

    \draw [yshift=-.5pt,very thick, decorate,decoration={brace,amplitude=5pt,mirror,raise=3.3ex}]
    	(-3,2.3) -- (-3,-2.3) node[midway,yshift=-3em]{};

\end{tikzpicture}
    \caption{Illustration of the ultrasuccinct TREX data structure on a directed graph $G$. The dashed edges form the spanning tree $T$ (computed via the MST-approximation algorithm from~\wref{sec:approximation}). The vertices are labelled as `new:old' where the new labels are induced by the level-order traversal of rooted tree $T$. The array~$D$ stores the direction of edges in~$T$. The array~$A'$ represents the adjacency list of~$G$ after removing~$T$ in the order of vertices appearing in level-order. The one-bits in~$S'$ mark the start indices of the outneighbourhood of vertices in $A'$.}
    \label{fig:ds-example}
\end{figure}
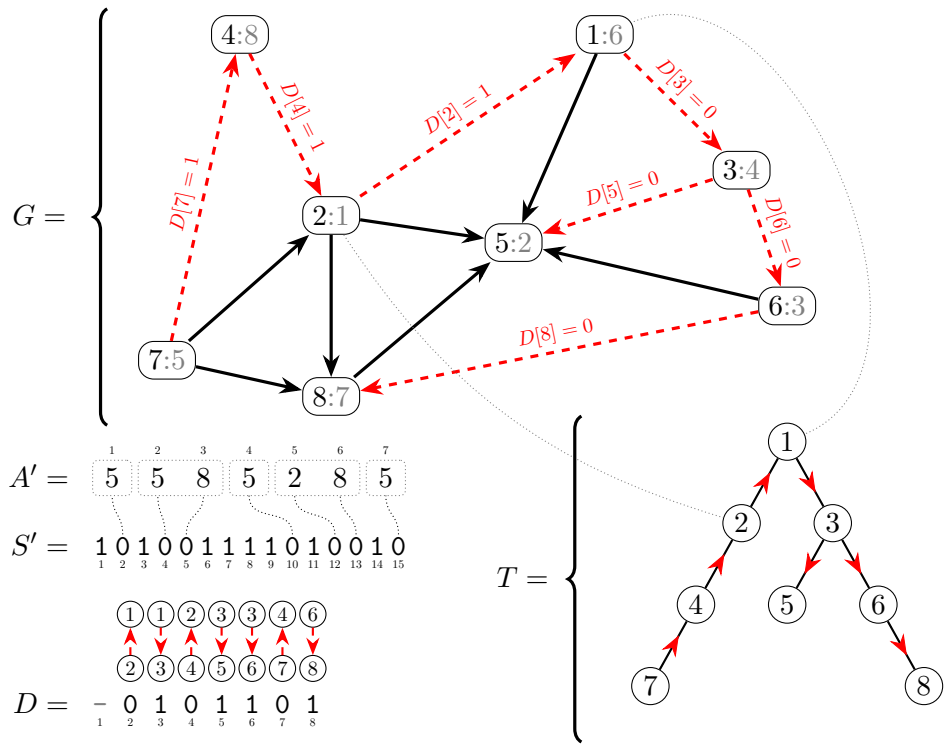

We combine this data structure with \emph{twin removal}: 
$u$ and $v$ are twins if their neighbourhoods coincide;
for an unlabelled graph, it suffices to remember the number of twins of each vertex instead of storing $u$ and $v$ separately.
TREX with twin removal yields an instance-optimal-space data structure for unlabelled graphs arising from the \emph{random copy model}~\cite{TurowskiMagnerSzpankowski2020} (see also \wref{app:copy-model}). 

\paragraph{Structure Over Identity}

We point out that the relabeling mapping used in tree extraction can be explicitly output at construction time, so a user of the data structure may choose to remember some of the new vertex labels to facilitate later queries.
However, the mapping is not stored as part of our TREX data structure.
We are thus storing an \emph{unlabelled graph}, \ie, only the equivalence class \wrt graph isomorphism instead of a concrete labelled graph. 
This reduces the intrinsic information content by up to $\lg(n!) \sim n \lg n$ bits, and these maximal savings are asymptotically attained for various distributions over graphs~\cite{ChoiSzpankowski2012,LuczakMagnerSzpankowski2019,PatonHartleStepanyantsVanDerHoorn2022,IsmailiAlaouiNamrataWild2025}.

Apart from the information-theoretic question of identifying the intrinsic information content of graph-structured data~\cite{SzpankowskiGrama2018} (in particular the notion of \emph{structural entropy}),
many applications can profit from these space savings.
For example, identifying \emph{network motifs}, \ie, small subgraphs that occur surprisingly often or surprisingly rarely in a network compared to random graphs, is a standard tool in network analysis~\cite{MiloShenOrrItzkovitzKashtanChklovskiiAlon2002}.
This involves counting occurrences of many motifs in many random graph chosen from a null model.
Here, vertex ids are randomly generated anyways, so renaming vertices does not change the result.

To further stress how widespread working with unlabelled graphs is,
we point out that the area of succinct graph data structures is full of examples
where original labels are not stored. 
Indeed, preserving the original vertex names would often dominate overall space
(\eg{} in trees~\cite{NavarroSadakane2014,FarzanMunro2014}, 
planar graphs~\cite{Jacobson1989,FerresFuentesSepulvedaGagieHeNavarro2020,KammerMeintrup2023}, 
proper interval graphs~\cite{AcanChakrabortyJoRao2020,HeMunroNekrichWildWu2020}, 
bipartite permutation graphs~\cite{TsakalidisWildZamaraev2023},
bounded clique-width graphs~\cite{Kamali2017,ChakrabortyJoSadakaneRao2024},
separable graphs~\cite{BlellochFarzan2010}) 
or affect the leading term of space usage
(\eg{} in interval graphs~\cite{AcanChakrabortyJoRao2020}, 
permutation graphs~\cite{TsakalidisWildZamaraev2023}
circle graphs~\cite{AcanChakrabortyJoNakashimaSadakaneRao2022}).
In all these cases, the known succinct data structures store unlabelled graphs,
often leaving the assumption implicit that renaming vertices must be acceptable.

\paragraph{Undirected Graphs}

If we start with an undirected graph, we have even more freedom:
We may extract a spanning tree for assigning vertex names, and also orient all edges
to minimise the resulting indegree entropy.
(Recall that wavelet trees already support navigation over both outgoing and incoming edges,
so storing the edge in one orientation is sufficient.)
Both our hardness results and our approximation algorithms directly generalise to this case.

\subsection{Related Work}

Compressed representations of graphs are in strong demand and 
corresponding data structures are an active area of research.
Much work in this area falls into one of two camps:

\emph{(1) Applied graph compression.}
Here, with properties of certain domains in mind, compression heuristics are developed and
evaluated empirically on benchmarks datasets.
A~representative example is the webgraph framework~\cite{BoldiVigna2004}, 
which provides space-efficient data structures for typical networks of websites and hyperlinks.
Many more are discussed in~\cite{BestaHoefler2019}.

\emph{(2) Succinct representation of a graph class.}
A graph class $\mathcal G$ is a set of graphs closed under isomorphism.
A succinct representation of a graph class uses $\lg|G_n|(1+o(1))$ bits of space,
where $G_n$ is the set of graphs with vertex set $[n]$ in $\mathcal G$.
In other words, we asymptotically use the optimal worst-case number of bits to represent graphs with a known number of vertices~$n$.
Succinct representations with efficient (often $O(1)$ time) navigational query support
are known for many graph classes; as nonexhaustive examples, we list
(proper) interval graphs~\cite{AcanChakrabortyJoRao2020,HeMunroNekrichWildWu2020}, 
(bipartite) permutation graphs~\cite{TsakalidisWildZamaraev2023},
bounded clique-width graphs~\cite{Kamali2017,ChakrabortyJoSadakaneRao2024},
separable graphs~\cite{BlellochFarzan2010},
circle graphs~\cite{AcanChakrabortyJoNakashimaSadakaneRao2022},
chordal graphs~\cite{MunroWu2018}.
Planar maps can also be succinctly represented~\cite{CastelliAleardiSchaeffer2008};
planar \emph{graphs} (without a fixed embedding) still pose challenges since not even the asymptotic number of unlabelled graphs on $n$ vertices is known, 
but representations with $O(\lg|G_n|)$ bits are achievable~\cite{FerresFuentesSepulvedaGagieHeNavarro2020,KammerMeintrup2023}.
All of the above works (implicitly) consider unlabelled graphs and assign new vertex ids at construction time.

A common property of the above succinct representations of a graph class is that they use the \emph{same} space for any graph of a given size, and are highly specific to the graph class at hand.
In general, succinct data structures aim to support efficient queries for an object $x\in \mathcal X$ using $\lg|\mathcal X|(1+o(1))$ bits of space. This space usage is asymptotically optimal in the \emph{worst case} over $\mathcal X$, but can, in 
principle, be improved to $\lg(1/\Prob x)$ bits of space when the object is drawn randomly from $\mathcal X$ with probability $\Prob x$. 
A much smaller number of prior works consider data structures whose space \emph{adapts} to the compressibility of the concrete instance; such data structures are known for trees~\cite{JanssonSadakaneSung2012,MunroNicholsonSeelbachBenknerWild2021} and preferential-attachment graphs~\cite{IsmailiAlaouiNamrataWild2025}.

We generalise the work of~\cite{IsmailiAlaouiNamrataWild2025} from the specific model of out-regular graphs generated by preferential attachment \emph{(Barabási-Albert model)} to arbitrary directed or undirected graphs. 
Unlike in~\cite{IsmailiAlaouiNamrataWild2025}, general graphs require a principled choice of the spanning tree to extract and a more elaborate analysis.  Especially for undirected graphs, we also consider the additional degree of freedom to choose an orientation for all edges.
Like in~\cite{IsmailiAlaouiNamrataWild2025}, our resulting algorithms are simple and use only standard tools from succinct data structures, for which efficient implementations are available.

GLOUDS~\cite{FischerPeters2016} is a precursor to~\cite{IsmailiAlaouiNamrataWild2025} and similar to our data structure in that they also partition the graph into a tree and the remaining edges.
However, the focus there is on a simple data structure. 
The extracted tree is always a BFS tree/forest and the remaining edges are encoded \emph{alongside} the LOUDS (level-order unary degree sequence) representation of the tree.
Non-tree edges are only compressed in that their target is stored using $\lceil\lg h\rceil$ bits where $h$ is the number of different vertices that appear as targets of the $k=m-(n-1)$ non-tree edges;
however, no attempt is made at choosing a tree to obtain a minimal $h$.
Storing edge targets in an entropy-compressed form is not pursued further.

The problem of orienting edges of an undirected graph to minimise the resulting indegree entropy
(without tree extraction) was considered by Cardinal et al.~\cite{CardinalFioriniJoret2008}.
They showed that the problem is $\NP$-hard even if the graph is planar.
On the other hand, simply orienting all edges towards the larger-degree endpoint yields 
an indegree entropy of at most $\mathit{OPT}+m$.
We can prove a similar result, but with weaker constants, via~\wref{thm:generic_apx}.

\subsection{Contributions}

In this section, we succinctly formalise our technical contributions.
Proofs are delayed to later sections.

Our first result is a generic way to obtain approximate algorithms for \emph{entropy minimisation} by transforming exact algorithms for \emph{linear minimisation}.
(This theorem may seem abstract at first, however it immediately implies that several
of the problems we are interested can be approximated up to a linear error term;
we illustrate this below on \TheOptProblem.)

\begin{restatable}{theorem}{genericapx}
\label{thm:generic_apx}
    Fix some vector $(d_1, \ldots, d_n) \in \mathbb{N}^n$ with $m = \sum_i d_i$.
    Fix some positive integer~$k$.
    Fix also some class
    \[
        \mathcal{C} \wrel\subseteq
        \left
        \{ (x_1, \ldots, x_n) \in \mathbb{N} \;\;\middle\vert\;\;
			\forall i \ 0 \leq x_i \leq d_n \wbin\wedge \displaystyle\sum_ix_i = k
        \right\}.
    \]
    Consider the following two objective functions:
    \begin{align*}
        \optlin(x_1, \ldots, x_n) &\wwrel= \sum_i x_i \lg d_i,\\
        \optent(x_1, \ldots, x_n) &\wwrel= -\sum_i (d_i - x_i) \lg \left(\frac{d_i - x_i}{m - k}\right)
        \wwrel= H(d_1 - x_1, \ldots, d_n - x_n).
    \end{align*}
    The following holds for any $(x_1, \ldots, x_n) \in \mathcal{C}$:
    \begin{align*}
    	&
        	\optent(x_1, \ldots, x_n) \bin- \frac{k}{\ln 2}
    \\	&\wwrel\leq 
        	(m - k) \lg(m - k) - \left(\sum_i d_i \lg d_i\right) 
        	\bin+ \optlin(x_1, \ldots, x_n)
    \\	&\wwrel\leq 
        	\optent(x_1, \ldots, x_n).
        \end{align*}
    Hence, if we can minimise $\optlin$ over $\mathcal{C}$ exactly, then we can
    minimise $\optent$ over $\mathcal{C}$ with additive error at most $k/\ln 2$.
\end{restatable}

Let us see how to apply this theorem to \TheOptProblem.
Suppose we are given a (weakly) connected digraph $G$ with $n$ vertices and $m$
edges. Let $d_1, \ldots, d_n$ be the indegree sequence of~$G$. Take $k = n - 1$,
and define $\mathcal{C}$ to be the set of indegree sequences of spanning trees of~$G$.\footnote{To clarify: we are interested in finding subsets of edges that are trees when ignoring edge orientations; this is what is meant by spanning tree.}
In other words, $\mathcal{C}$ contains, for every spanning tree $T$ of $G$, the sequence
$x^{(T)} = (x_1^{(T)}, \ldots, x_n^{(T)})$ where
\[
    x_i^{(T)} \wrel= \bigl| \{ (j, i) \in T \} \bigr|.
\]
Notice that for any spanning tree $T$ of $G$, we have that the empirical entropy
of the indegree sequence of $G - T$ is, by definition, $\optent(x^{(T)})$. Hence,
the \TheOptProblem problem is identical to minimising
$\optent$ over $\mathcal{C}$.%
\footnote{%
	Strictly speaking, we need to return $T$ not the vector $x^{(T)}$, 
	but this is also easy to compute.%
} 
In view of \wref{thm:generic_apx},
this can be done with error $(n - 1) / \ln 2$, provided we can minimise $\optlin$ over $\mathcal{C}$.
But, expanding definitions, this is just finding a minimum cost spanning tree where
the cost of edge $(i, j)$ is given by $\lg d_j$. This can be done easily~-- hence
the following key corollary follows.

\begin{corollary}\label{cor:MINETREX_apx}
    \TheOptProblem can be approximated with additive error $(n - 1) / \ln 2$ in $O(n + m)$ time.
\end{corollary}
\begin{proof}
    The previous discussion implies that it suffices to compute a
    minimum cost spanning tree where the cost of an edge $(i, j)$
    is $\lg(d_j)$, where $d_j$ is the indegree of $j$. To do this
    in linear time, observe that we can just as well assign edge
    $(i, j)$ cost $d_j$ without modifying the minimum cost spanning
    tree.
    Then, since the new costs are integers between $1$ and $m$,
    we can easily find the desired spanning tree using Prim's
    algorithm in $O(n + m)$ time with a bucket-based priority queue.
\end{proof}

In a similar way (see \wref{sec:approximation} for proof), 
we can deal with selecting both the orientation of edges on top of tree extraction.
We denote this problem as \textsc{Undirected} \TheOptProblemLongForm (U-\TheOptProblem).

\begin{restatable}{corollary}{UMINETREXAPX}\label{cor:U_MINETREX_apx}
    U-\TheOptProblem can be approximated with additive error $(m + n - 1) / \ln 2$ in $O(n + m)$ time.
\end{restatable}

\begin{figure}
    \centering
    \begin{tikzpicture}[digraph, scale=1.4,shorten >=1pt]
        \begin{scope}[local bounding box=inst]
            \node (A) at (0:1cm){};
            \node (B) at (72:1cm){};
            \node (C) at (144:1cm){};
            \node (D) at (-144:1cm){};
            \node (E) at (-72:1cm){};
            \node (X) at (3, 0){};
            \node (P1) at ($(X) + (150:1cm)$){ };
            \node (P2) at ($(X) + (120:1cm)$){ };
            \node (P3) at ($(X) + (90:1cm)$){ };
            \node (P4) at ($(X) + (60:1cm)$){ };
            \node (P5) at ($(X) + (30:1cm)$){ };
            \node (P6) at ($(X) + (0:1cm)$){ };
            \node (P7) at ($(X) + (-30:1cm)$){ };
            \node (P8) at ($(X) + (-60:1cm)$){ };
            \node (P9) at ($(X) + (-90:1cm)$){ };
            \node (P10) at ($(X) + (-120:1cm)$){ };
            \node (P11) at ($(X) + (-150:1cm)$){ };
            \draw (P1) -> (X);
            \draw (P2) -> (X);
            \draw (P3) -> (X);
            \draw (P4) -> (X);
            \draw (P5) -> (X);
            \draw (P6) -> (X);
            \draw (P7) -> (X);
            \draw (P8) -> (X);
            \draw (P9) -> (X);
            \draw (P10) -> (X);
            \draw (P11) -> (X);
            \draw (A) -> (X);
		\begin{scope}[bend right=10]
            \draw (A) to (B);
            \draw (B) to (C);
            \draw (C) to (D);
            \draw (D) to (E);
            \draw (E) to (A);
            \draw (B) to (A);
            \draw (C) to (B);
            \draw (D) to (C);
            \draw (E) to (D);
            \draw (A) to (E);
            \draw (A) to (C);
            \draw (B) to (D);
            \draw (C) to (E);
            \draw (D) to (A);
            \draw (E) to (B);
            \draw (C) to (A);
            \draw (D) to (B);
            \draw (E) to (C);
            \draw (A) to (D);
            \draw (B) to (E);
        \end{scope}
        \end{scope}
    \end{tikzpicture}
    \caption{%
    	Example of a graph where tree extraction does not lower entropy. 
    	The clique on the left has $\sqrt n$ vertices in general.
    }
    \label{fig:badGraph}
\end{figure}
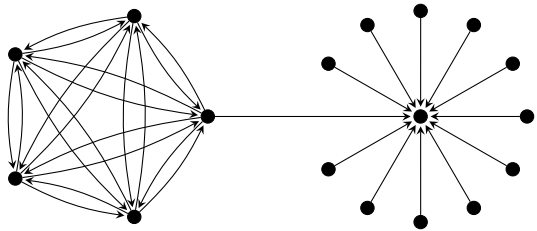

Let us now move on to hardness. 
We prove that for polytime algorithms, we will have to content
ourselves with a linear error term unless $\P=\NP$,
thus resolving the approximability of (U-)\TheOptProblem up to the first order.

\begin{theorem}\label{thm:MINETREX_hard}
    There exists some $\delta>0$ such that it is \NP-hard to approximate MINETREX
within additive error $\delta n$.
\end{theorem}

\begin{theorem}\label{thm:U_MINETREX_hard}
    There exists some $\delta>0$ such that it is \NP-hard to approximate U-MINETREX
within additive error $\delta m$.
\end{theorem}

These theorems follow by a reduction from a hardness result by Guruswami and Trevisan~\cite{GuruswamiTrevisan2005} on \textsc{Gap Exact Hitting Set}. 

While we find \TheOptProblemLongForm interesting in its own right,
from the application perspective of compressing graphs, 
finding the best tree to extract would be utterly useless if either 
(i) every tree is more or less equally good, or 
(ii) \emph{no} tree leads to a significant savings whatsoever.

It is easy to see that (i) is not true in general. Indeed, consider \wref{fig:goodGraph}, or more generally
a graph with $2n + 2$ vertices, labelled $a_1, \ldots, a_n, \, b_1, \ldots, b_n, \, x, y$, with edges
$a_1 \to \cdots \to a_n \to b_1 \to \cdots \to b_n$, $a_i \to x$, $b_i \to y$ and $x \to y$.
In such a graph, if we remove the path	
\[
    y \leftarrow x \leftarrow a_1 \to \cdots \to a_n \to b_1 \to\cdots\to b_n,
\]
we are left with the edges $a_2 \to x, \ldots, a_n \to x$ and $b_1 \to y, \ldots, b_n \to y$ (two disjoint stars). 
Our scheme stores the resulting graph using $O(n)$ bits (using 1 bit per edge in $G-T$, for a total of $\sim 2n$ bits plus neighbourhood boundaries for $G-T$). 
On the other hand, if we remove the edges
$a_1 \to x, \ldots, a_n \to x$, $b_1 \to y, \ldots, b_n \to y$ and $x \to y$, then we are left with
the path $a_1 \to \cdots \to a_n \to b_1 \to \cdots \to b_n$, 
with indegree entropy $\Omega(\lg n)$ bit (as for the original graph if we do not remove anything). 
Thus we can achieve a significant space saving using our
approximation scheme~-- and indeed, the error of $O(n)$ is negligible compared to that saving.

Question (ii) is rather more thorny. It turns out that there are cases where \emph{no} tree 
significantly lowers the entropy if deleted. Consider, for example, a graph that looks like the graph from
\wref{fig:badGraph}. The general pattern shall have a clique on the left with $k$ vertices and a star on the right with $k^2$ vertices; so $k \sim \sqrt n$.
In this example, the original graph has $\DegreeEntropy(G) \sim \frac12\lg k \sim \frac14 \lg n$, for a total of
$\Theta(n \lg n)$ bits to store $G$. 
The problem is that the star on the right side of the graph costs almost nothing
to store~-- most of the complexity of the graph is contained in the small but dense part on the left.
This dense part is also almost completely untouched by removing any tree~-- we see that, indeed,
any tree $T$ that we remove only affects a vanishingly small fraction of the edges in the clique,
leaving $\DegreeEntropy(G-T)\sim \frac14 \lg n$.

The essential reason why \wref{fig:badGraph} has this bad behaviour is that its density is completely
contained within a small part: the clique was much denser than the whole graph. This is problematic
because the tree cannot lower the amount of bits needed to store the adjacency information in the
dense part of the graph very much. On the other hand, the vast majority of the graph is
nearly free to store, and thus we do not profit much from removing it. We show that this is essentially
the \emph{only} case where this happens: for any digraph $G$ where no subgraph is significantly denser than
the entire graph, we show that we can always achieve a significant saving via tree extraction.

In the following, for a graph $G$ with $n$ vertices and $m$ edges, define its density $\delta(G) = m / n$.

\begin{restatable}{theorem}{thmEquanimity}\label{thm:equanimity}
    Consider some connected $n$-vertex $m$-edge graph $G = (V, E)$. Let $\alpha$ be some constant so that
for any subgraph $G'$ of $G$, we have $\delta(G') \leq \alpha \delta(G)$. Then, there exists some
spanning tree $T$ of $G$ such that
\[
    H(G - T) \wwrel\leq H(G) - \frac{n}{2\alpha} \DegreeEntropy(G) + \frac{2n}{\ln 2}.
\]
\end{restatable}

For example, if $G$ has per-edge entropy say $\frac{1}{3} \lg n$, and no subgraph has \emph{proportionally}
more than 2 times the number of edges as the entire graph, then we can save at least roughly
$\frac{1}{12} n \lg n - 2n / \ln 2$ bits. 
While we make no claim about whether the bound in \wref{thm:equanimity} is tight, 
note that if the graph has, \eg, constant per-edge entropy \DegreeEntropy, there is no hope
to save more than a linear number of bits.

Additionally, we show that in any digraph where a significant number of vertices have nonzero indegree,
we also achieve a significant savings. 

\begin{restatable}{theorem}{thmManyNonzero}\label{thm:manyNonzero}
    Consider some connected $n$-vertex $m$-edge graph $G = (V, E)$, where at least $n / \alpha$ vertices have
nonzero indegree. Then there exists some spanning tree $T$ of $G$ such that
\[
    H(G - T) \wwrel\leq H(G) - \frac{n}{2\alpha} \DegreeEntropy(G) + \frac{2n}{\ln 2}.
\]
\end{restatable}
For example, if $G$ has per-edge entropy say $\frac{1}{3} \lg n$, and at least $1 / 3$ of the vertices
have nonzero indegree, then we can save roughly $\frac{1}{18 } n \lg n - 2n / \ln 2$ bits.

Finally, we provide a data structure capable of storing a TREXed unlabelled graph while supporting all operations a classical adjacency list representation offers (plus navigation to inneighbours) in $O(\lg n)$ time (see also \wref{fig:ds-example}).

\begin{restatable}[Ultrasuccinct TREX]{theorem}{ultrasuccincttrex}
\label{thm:ultrasuccinct-trex-directed}
	Given a directed graph~$G$ and spanning tree $T$ in~$G$, we can construct a data structure for a graph isomorphic to $G$ (where the vertex relabeling can be output at construction time) that occupies 
    $H(G-T) + 3n + o(m) + \lg \binom{m}{n} \le 
    H(G-T) + n\lg(m/n) + 4.4427n + o(m)$ bits of space, while supporting the operations $\mathsf{N_{out}}(v, i)$ ($i$th outneighbour of $v$), $\mathsf{N_{in}}(v, i)$ ($i$th inneighbour of $v$),  $\mathsf{indegree}(v)$, and 
    $\mathsf{adjacent}(u, v)$ in $O(\lg n)$ time;
    $\mathsf{outdegree}(v)$ takes $O(1)$ time.
\end{restatable}

Lastly, we show how to combine the above data structure
with \emph{twin removal}, a further compression method for unlabelled graphs where we identify all vertices with identical neighbourhood.
Details are given in \wref{sec:twin-removal}.

\subsection{Hardness of exactly solving MINETREX.}
\label{sec:exact-hardness}

While we defer the proof of \wref{thm:MINETREX_hard} and \wref{thm:U_MINETREX_hard} to \wref{sec:hardness}, to get a taste of the reduction, we informally sketch the key idea
for the proof of the simpler statement that \emph{exact} MINETREX is \NP-hard.

Recall that in the \textsc{Exact Cover By 3-Regular 3-Sets (X3C)} problem, we are given a universe $\{1, \ldots, 3n\}$, and a family of sets $S_1, \ldots, S_{3n} \subseteq [3n]$, such that every set has size 3,
and every element belongs to exactly 3 sets.
The goal is then to select $n$ sets so that every element in the universe is
contained within precisely one of the selected sets, or to determine that this
is impossible.

This problem is well-known to be $\NP$-hard~\cite{Gonzalez1985}, and we can show the $\NP$-hardness
of exactly solving MINETREX by reducing from this problem. It will be, however, more
convenient to formulate X3C in a graph-theoretic way for the purpose of this reduction.
Observe that we can represent our instance by a bipartite graph of vertex set $A \cup B$ (the two bipartitions), where $A$ represents the sets, and $B$ represents the universe. We add an edge $u \to v$ between
the vertex $u \in B$ representing $i$ and the vertex $v \in A$ representing $S_j$ if and only if $i \in S_j$.
The goal is now to select $n$ of the vertices in $A$, so that every
vertex in $B$ has exactly one selected neighbour.

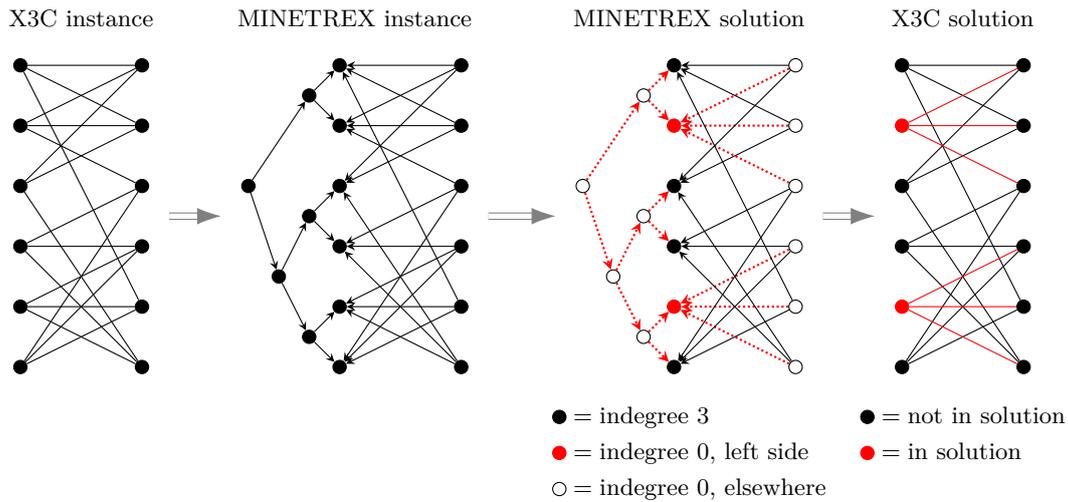
\begin{figure}
\centering
\begin{tikzpicture}[scale=0.8]

\newcommand\DrawRightPoints{
    \node (B0) at (2, 0){};
    \node (B1) at (2, 1){};
    \node (B2) at (2, 2){};
    \node (B3) at (2, 3){};
    \node (B4) at (2, 4){};
    \node (B5) at (2, 5){};
}
\newcommand\DrawLeftChosenPoints{
    \node (A1) at (0, 1){};
    \node (A4) at (0, 4){};
}
\newcommand\DrawLeftUnchosenPoints{
    \node (A0) at (0, 0){};
    \node (A2) at (0, 2){};
    \node (A3) at (0, 3){};
    \node (A5) at (0, 5){};
}
\newcommand\DrawGraphPoints{
    \DrawLeftChosenPoints{}
    \DrawLeftUnchosenPoints{}
    \DrawRightPoints{}
}
\newcommand\DrawTreePoints{
    \node (X) at (-.5, 0.5){};
    \node (Y) at (-.5, 2.5){};
    \node (Z) at (-.5, 4.5){};
    \node (S) at (-1, 1.5){};
    \node (T) at (-1.5, 3){};
}

\newcommand\DrawChosenGraphEdges{
    \draw (B0) -> (A1);
    \draw (B1) -> (A1);
    \draw (B2) -> (A1);
    \draw (B3) -> (A4);
    \draw (B4) -> (A4);
    \draw (B5) -> (A4);
}

\newcommand\DrawUnchosenGraphEdges{
    \draw (B5) -> (A5);
    \draw (B0) -> (A3);
    \draw (B0) -> (A2);
    \draw (B1) -> (A0);
    \draw (B1) -> (A5);
    \draw (B2) -> (A2);
    \draw (B2) -> (A0);
    \draw (B3) -> (A0);
    \draw (B3) -> (A2);
    \draw (B4) -> (A5);
    \draw (B4) -> (A3);
    \draw (B5) -> (A3);
}

\newcommand\DrawGraphEdges{
    \DrawChosenGraphEdges{}
    \DrawUnchosenGraphEdges{}
}

\newcommand\DrawTreeEdges{
    \draw (X) -> (A0);
    \draw (X) -> (A1);
    \draw (Y) -> (A2);
    \draw (Y) -> (A3);
    \draw (Z) -> (A4);
    \draw (Z) -> (A5);
    \draw (S) -> (X);
    \draw (S) -> (Y);
    \draw (T) -> (Z);
    \draw (T) -> (S);
}
\scoped[graph, local bounding box=init_inst]{
    \DrawGraphPoints{}
    \DrawGraphEdges{}
}

\scoped[digraph, xshift=5.25cm, local bounding box=red_inst]{
    \DrawGraphPoints{}
    \DrawTreePoints{}
    \DrawGraphEdges{}
    \DrawTreeEdges{}
}

\scoped[digraph, xshift=10.75cm, local bounding box=red_sol]{
    \scoped[every node/.append style={fill=white}]{
        \DrawTreePoints{}
        \DrawRightPoints{}
    }
    
    \scoped[every node/.append style={red, fill=red}]{
        \DrawLeftChosenPoints{}
    }
    
    \DrawLeftUnchosenPoints{}
    
    \scoped[chosen edge]{
        \DrawChosenGraphEdges{}
        \DrawTreeEdges{}
    }
    \DrawUnchosenGraphEdges{}
}

\scoped[graph, xshift = 14.5cm, local bounding box=init_sol] {
    \scoped[every node/.append style={red, fill=red}]{
        \DrawLeftChosenPoints{}
    }
    
    \DrawLeftUnchosenPoints{}
    \DrawRightPoints{}
    \DrawUnchosenGraphEdges{}
    
    \scoped[every path/.append style={red}]{
        \DrawChosenGraphEdges{}
    }
}

\node[above=0.3cm of init_inst.north] {\footnotesize X3C instance};

\node[above=0.3cm of red_inst.north] {\footnotesize MINETREX instance};

\node[above=0.3cm of red_sol.north] {\footnotesize MINETREX solution};

\node[above=0.3cm of init_sol.north] {\footnotesize X3C solution};

\node[below=0.3cm of red_sol.south, align=left]{\footnotesize $\vcenter{\hbox{\tikz{\node[draw,circle,minimum size=5pt,inner sep=0pt,fill=black]at (0, 0){};}}} = \text{indegree 3}$\\
\footnotesize$\vcenter{\hbox{\tikz{\node[draw=red,circle,minimum size=5pt,inner sep=0pt,fill=red]at (0, 0){};}}} = \text{indegree 0, left side}$\\
\footnotesize$\vcenter{\hbox{\tikz{\node[draw,circle,minimum size=5pt,inner sep=0pt,fill=white]at (0, 0){};}}} = \text{indegree 0, elsewhere}$
};

\node[below=0.3cm of init_sol.south, align=left]{ \footnotesize$\vcenter{\hbox{\tikz{\node[draw=black,circle,minimum size=5pt,inner sep=0pt,fill=black]at (0, 0){};}}} = \text{not in solution}$\\
\footnotesize $\vcenter{\hbox{\tikz{\node[draw=red,circle,minimum size=5pt,inner sep=0pt,fill=red]at (0, 0){};}}} = \text{in solution}$};

\node[right=0cm of init_inst.east] (Q){};
\node[left=0cm of red_inst.west] (W){};

\node[right=0cm of red_inst.east] (E){};
\node[left=0cm of red_sol.west] (R){};

\node[right=0cm of red_sol.east] (T){};
\node[left=0cm of init_sol.west] (Y){};

\draw[macro edge] (Q) -> (W);
\draw[macro edge] (E) -> (R);
\draw[macro edge] (T) -> (Y);
\end{tikzpicture}
\caption{Reduction of X3C to MINETREX (the left-hand-side bipartition represents the subcollections, and the right-hand-side bipartition, the universe set).}\label{fig:thing}
\end{figure}

We picture our reduction in \wref{fig:thing}. We hang an arbitrary binary tree on the left part of the graph such that $A$ forms its leaves, where all edges are oriented towards the leaves. To better understand the properties of entropy, we recall the following lemma.
Recall $H(d_1,\ldots,d_n) = \sum_{v\in V} d_v \log_2(m/d_v)$.

\begin{lemma}[{{Domination, see \eg, \cite[Lemma 1]{CardinalFioriniJoret2008}}}]\label{lem:domination}
    Suppose that $x=(x_1, \ldots, x_n)$ and $y=(y_1, \ldots, y_n)$ are sequences of non-negative real numbers with $x_1+\cdots+x_n = y_1+\cdots+y_n$, and let $x_1 \geq \cdots \geq x_n$ and $y_1 \geq \cdots \geq y_n$. 
    We say $x$ \emph{dominates} $y$ if $x_1 + \cdots + x_i \geq y_1 + \cdots + y_i$ for all $i \in [n]$,
    and $x$ \emph{strictly dominates} $y$ if additionally $x\ne y$.
    
    If $x$ strictly dominates $y$, then $H(x_1, \ldots, x_n) < H(y_1, \ldots, y_n)$.
\end{lemma}

In particular, note that the total number of edges in any subgraph of our instance resulting from the elimination of a spanning tree is the same. Furthermore, note that the maximum indegree of any vertex is 4, and the vertices with indegree 4 have outdegree 0, and thus after eliminating a spanning tree have degree at most 3. Thus the most dominant (in the sense of \wref{lem:domination}) sequence of indegrees has form $(3, \ldots, 3, 0, \ldots, 0)$. The only way this can happen is if (i) we extract every edge completely within the added tree, and (ii) every vertex in the left hand side has either exactly 1 or exactly 4 edges extracted. Suppose $S$ is the subset of vertices on the left side where we take 4 neighbouring edges. Note that since the subtree we extract is a spanning tree, it must be the case that the root of the tree we added in the reduction is connected to every vertex on the right side. This can only happen via the vertices in $S$; in particular, by a simple calculation we see that every vertex on the right has exactly one vertex in $S$ as a neighbour.

Thus, we have shown that the optimal possible solution to the 
reduced MINETREX instance \emph{must}
correspond to a solution to the
original X3C instance, and vice
versa. This implies the hardness of exactly solving
MINETREX (and the proof for U-MINETREX is quite similar). Our proofs for \wref{thm:MINETREX_hard} and \wref{thm:U_MINETREX_hard} use essentially the same reduction, but starting from a gap version of X3C (actually exact cover with different
regularity parameters, not 3). The soundness analysis of our
reduction is also more involved,
since we must show that we can
decode an approximate solution of (U-)MINETREX.

\subsection{Open Problems}

Since all components of tree extraction are relatively simple,
it is conceivable to work well in practice.
We reserve a proper empirical investigation of the potential of tree extraction for a future work.

While we focused our attention entirely on trees,
the basic idea of tree extraction easily extends beyond trees:
extract some subgraph $G'$ from $G$ that can be stored using a constant number of bits per edge and use the (implicit) relabeling for $G-G'$; if we additionally can efficiently find a maximum size subgraph (in terms of number of edges),
we may further improve the space savings.

\paragraph{Outline}
The remainder of this paper is organised as follows.
\wref{sec:approximation} describes our MINETREX approximation algorithm and the generic analysis framework for entropy minimisation.
\wref{sec:hardness} gives the inapproximability reductions.
\wref{sec:data-structures} describes our ultrasuccinct TREX data structure that adds efficient query support to unlabelled graphs;
\wref{sec:twin-removal} extends this construction by first removing twins in a graph.
In \wref{sec:good-graphs}, we prove that large classes of graphs indeed allow for significant space savings via tree extraction.
In the appendix, we give details of the used \NP-hard gap problem from~\cite{GuruswamiTrevisan2005} (\wref{app:gt05}), show that there are always $\subseteq$-maximal optimal solutions for MINETREX (\wref{app:tree-not-forest}), and discuss instance-optimal compression for the copy model (\wref{app:copy-model}).

\section{Greedy Approximation}
\label{sec:approximation}
\label{sec:greedy}

In this section we prove \wref{thm:generic_apx}, as well as \wref{cor:U_MINETREX_apx}, thus completing our positive results for (U-)MINETREX.
We need the following easy lemma.
\begin{lemma}\label{lem:simpleBound}
    For $d \geq x \geq 0$, we have $(d - x) (\lg d - \lg (d - x)) \leq \frac{x}{\ln 2}$.
\end{lemma}
\begin{proof}
Recall that $t + 1 \leq e^t$ for all $t$.
There is nothing to prove for $d = x$, so assume $d > x$.
Let $t = \ln(d / (d - x)) \geq 0$. We then find that
\[
\ln d - \ln(d - x) \wwrel\leq \frac{d}{d - x} - 1,
\]
which, after rearranging and dividing by $\ln 2$, implies our result.
\end{proof}

We restate \wref{thm:generic_apx} for the reader's convenience.

\genericapx*

\begin{proof}
    First we apply some standard transformations on the entropy: note that
    \begin{equation}\label{eq:optent_eq}
    \begin{aligned}[b]
    		\optent(x_1, \ldots, x_n)
    	&\wwrel= 
    		-\sum_i (d_i - x_i) \lg\left(\frac{d_i - x_i}{m - k}\right)
    \\	&\wwrel=
         	\sum_i (d_i - x_i) \log(m - k)
        	\bin- \sum_i (d_i - x_i) \lg (d_i - x_i)
    \\	&\wwrel= 
    		(m - k) \log(m- k)
        	\bin- \sum_i (d_i - x_i) \lg (d_i - x_i).
    \end{aligned}
    \end{equation}
    By \wref{lem:simpleBound}, and by simple monotonicity, for every $i \in [n]$ we find that
    \[
    0 \wwrel\leq (d_i - x_i) \left(\lg d_i - \lg (d_i - x_i)\right) \wwrel\leq \frac{x_i}{\ln 2}.
    \]
    By summation over all $i$, and since $\sum_i x_i = k$, we have
    \[
    0 \wwrel\leq \sum_i (d_i - x_i) \left(\lg d_i - \lg (d_i - x_i)\right) \wwrel\leq \frac{k}{\ln 2},
    \]
    or equivalently
    \[
    - \sum_i (d_i - x_i) \lg d_i
    \wwrel\leq
    - \sum_i (d_i - x_i) \lg (d_i - x_i)
    \wwrel\leq
    \frac{k}{\ln 2} \bin- \sum_i (d_i - x_i) \lg d_i.
    \]
    Finally, we can recognise $\optlin$ in this expression to deduce
    \begin{align*}
    - \left(\sum_i d_i \lg d_i\right) + \optlin(x_1, \ldots, x_n)
    &\wwrel\leq
    - \sum_i (d_i - x_i) \lg (d_i - x_i)
    \\&\wwrel\leq
    \frac{k}{\ln 2}  - \left(\sum_i d_i \lg d_i\right) + \optlin(x_1, \ldots, x_n).
    \end{align*}
    Apply this to the expression found in~\eqref{eq:optent_eq} to find that
    \begin{equation}\label{eq:result}
    \begin{aligned}[b]
    	&
    		(m - k) \log(m - k) \bin- \sum_i d_i \lg d_i \bin+ \optlin(x_1, \ldots, x_n)
    \\	&\wwrel\leq 
    		\optent(x_1, \ldots, x_n)
    \\	&\wwrel\leq
    		\frac{k}{\ln 2} \bin+ (m - k) \log(m - k) 
    		\bin- \sum_i d_i \lg d_i \bin+ \optlin(x_1, \ldots, x_n),
    \end{aligned}
    \end{equation}
    which is exactly the bound we wanted to prove.

    Now, let us show that minimising $\optlin$ over $\mathcal{C}$ exactly implies being able to minimise
    $\optent$ over $\mathcal{C}$ with additive error at most $k / \ln 2$. For what follows, let
    \begin{align*}
    C &\wwrel= (m - k) \lg (m - k) - \left(\sum_i d_i \lg d_i \right)\\
    x^* &\wwrel= \arg \min_{x \in \mathcal{C}} \optent(x) \\
    \hat{x} &\wwrel= \arg \min_{x \in \mathcal{C}} \optlin(x).
    \end{align*}
    It suffices to prove that $\optent(\hat{x}) \leq \optent(x^*) + k / \ln 2$. But indeed,
    \begin{align*}
    \optent(\hat{x})
    & \wwrel\leq \frac{k}{\ln 2} + C + \optlin(\hat{x})  &&\text{\eqref{eq:result}} \\
    & \wwrel\leq \frac{k}{\ln 2} + C + \optlin(x^*) && \text{(Optimality of $\hat{x}$)} \\
    & \wwrel\leq \frac{k}{\ln 2} + \optent(x^*) && \text{\eqref{eq:result}}\iflipics{\qedhere}{}
    \end{align*}
\end{proof}

Next, let us prove that we can solve U-MINETREX with additive error $(m + n - 1) / \ln 2$.

\UMINETREXAPX*
\begin{proof}
    In U-MINETREX, we are given an undirected graph $G$ and must (i) orient every edge, and
(ii) delete a spanning subtree so that the indegree entropy of the resulting graph is as small as possible.
An equivalent formulation is the following: given a \emph{directed} graph $G$ where the edges
come in pairs of form $\{ (i, j), (j, i) \}$, (i) delete one edge from every pair, and (ii) delete
a spanning subtree from the resulting graph so that the indegree entropy of the resulting graph
is as small as possible.

This problem now readily falls within the framework of problems covered by \wref{thm:generic_apx}.
Let $d_1, \ldots, d_n$ be the initial indegrees of the vertices of $G$. Let $\mathcal{G}$ denote
the set of subgraphs of $G$ which consist of (i) exactly one edge from each pair $\{ (i, j), (j, i) \}$
in $G$, together with (ii) a disjoint spanning tree of $G$.
Let $\mathcal{C}$ consist of the indegree sequences of graphs from $\mathcal{G}$. Note that every such graph
contains exactly $m + n - 1$ edges, and hence the sum of every vector in $\mathcal{C}$ is $m + n - 1$. U-MINETREX is now minimising $\optent$ over $\mathcal{C}$.\footnote{To be precise, we must also
find the subgraph $H \in \mathcal{G}$ that gives rise to the minimising indegree sequence, but
this is easy to do in our case.}

Applying \wref{thm:generic_apx}, we see that in order to solve U-MINETREX with additive error $(m + n - 1) / \ln 2$,
it is sufficient to find a graph from $\mathcal{G}$ with minimum cost, where the cost of an edge
$(i, j)$ is given by $\lg d_j$. Consider a subgraph $H \in \mathcal{G}$ and pair of edges $\{ (i, j), (j, i) \}$. In the case that exactly one of $(i, j)$ and $(j, i)$ belongs to $H$, it is
easy to see that the graph otherwise agreeing with $H$ but containing the other edge among $(i, j)$ and
$(j, i)$ also belongs to $\mathcal{G}$. Since at least one of these edges must belong to $H$,
we see that in the graph that minimises $\optlin$, we always take the edge with
smaller cost from among these two (and perhaps the other one).

We thus see that it is safe to delete the lower cost edge from the pair $\{ (i, j), (j, i) \}$~-- that is,
the edge pointing towards $\arg \min_{k \in \{i, j\}} d_k$. After
this, all that there remains to do is to find a minimum cost spanning tree in the remaining graph.
As detailed in the proof of \wref{cor:MINETREX_apx}, this can be done in $O(n + m)$ time.
\end{proof}

\section{Hardness of Approximation}
\label{sec:hardness}

In this section, we show, if $\P \not= \NP$, there exists a constant $\delta > 0$ such that {(\textsc{U-})}\TheOptProblem is impossible to approximate in polynomial time up to an additive error of $\delta n$ bits, where $n$ is the number of vertices in the graph instance. To do so, we aid ourselves with two preliminary problems which we show to be $\NP$-hard; the second will be used as a reduction tool for {(\textsc{U-})}\TheOptProblem. We define those problems below.

First, a regular version of gap exact hitting set. Often exact hitting set is formulated as a constraint satisfaction problem: we are given a set of Boolean variables, and a set of constraints on those variables. A constraint is satisfied if and only if exactly one of the variables it is applied to is set to true, and the rest to false (``1-in-$k$ SAT''). 
We prefer, for linguistic convenience, an equivalent formulation based on bipartite graphs. Transform the variables from the previous formulation into the vertices on the left side of a graph, and the constraints into the vertices on the right side; add edges for all variable-constraint incidences.

We are thus given a bipartite graph $(A \cup B, E)$, and our goal will be to ``hit'' vertices from $B$ by choosing some vertices from $A$. Our goal is to hit as many vertices in $B$ as possible; however, a vertex is only ``hit'' if \emph{exactly one} of its neighbours is in the selected solution (cf.\ ``exact hitting set''). 
Note that $(r, k)$-biregularity of the bipartite graph corresponds to each variable in the constraints satisfaction instance being contained within exactly $r$ constraints, and each constraint containing exactly $k$ variables.
Note further, that a perfect solution (where all of $B$ is hit) corresponds to an \emph{exact cover} of $B$;
$A$ here is a collection of $r$-subsets of $B$. This (arguably more widely known) formulation suffices for the non-approximate hardness (cf.\ \wref{sec:exact-hardness}).

\begin{definition}
    In the \thmemph{$\alpha$-approximate $(r, k)$-biregular exact hitting set} problem (abbreviated $(\alpha, r, k)$-APX-EHS), we are given an undirected $(r, k)$-biregular bipartite graph $G = (A \cup B, E)$. A prospective solution
    to this problem is a subset $S \subseteq A$. We define the \thmemph{value $\xi(S)$} of a solution $S$ to be the fraction of vertices in $B$ which have exactly one neighbour in $S$.
    The goal of the problem is to decide whether there exists a solution with value 1, or whether none exists even with value at least $1 / \alpha$.
\end{definition}

\begin{prob}
  \probtitle{$(\alpha, r, k)$-APX-EHS}
  \probrow{Input}{$G = (A \cup B, E)$, a $(r, k)$-biregular bipartite graph.}
  \probrow{Solution}{A subset $S \subseteq A$.}
  \probrow{Value}{$\displaystyle\xi(S) = \frac{1}{|B|}
        | \{ b \in B : | S \cap N(b) | = 1 \} |$.}
  \probrow{Yes-inst.}{There exists a solution $S$ with $\xi(S) = 1$.}
  \probrow{No-inst.}{All solutions $S$ have $\xi(S) < 1 / \alpha$.}
\end{prob}

Reverting to the constraint satisfaction view, in this problem we are given an instance, and must decide whether it is perfectly solvable, or we cannot even satisfy a $1 / \alpha$ fraction of the constraints in the instance.
Guruswami and Trevisan~\cite{GuruswamiTrevisan2005} prove that a variant of this problem, namely one where the $r$-regularity of $A$ is omitted, is \NP-hard for every $1 \leq \alpha < e$ and for some $k$ depending only on $\alpha$. A careful reading of their proof shows that their reduction only produces instances with this $r$-regularity property, where $r$ again depends only on $\alpha$. Hence we deduce the following.

\begin{theorem}[{\cite{GuruswamiTrevisan2005}, Theorem 10}]\label{thm:base}
    Let $e$ denote the base of the natural logarithm.
    For any $\alpha < e$, there exist constants $r, k$ such that $(\alpha, r, k)$-APX-EHS is NP-hard.
\end{theorem}

For completeness, we walk through the proof of~\cite{GuruswamiTrevisan2005} in \wref{app:gt05}, highlighting the fact that the output instance is $r$-regular for some $r$ that depends only on $\alpha$.

\begin{corollary}
    There exist integers $r_2 \geq 10 k_2$ such that $(2, r_2, k_2)$-APX-EHS is NP-hard.
\end{corollary}
\begin{proof}
    First, note that integers $r, k$ such that this is true exist by \wref{thm:base}; it is not immediate however that $r \geq 10k$. However, fix some arbitrary $d$, and observe that if $G = (A \cup B, E)$ is an $(r, k)$-biregular bipartite graph, an input to $(\alpha, r, k)$-APX-EHS, then the graph $G' = (A \cup ([d] \times B), F)$
    where
    $F = \{ (a, (i, b)) \mid (a, b) \in E, i \in [d] \}$
    is an instance equivalent to $G$. This is because the set of solutions is the same, and furthermore (as can be seen by a simple calculation) the value of every solution is the same.

    Hence, we see that $(\alpha, r, k)$-APX-EHS reduces to $(\alpha, dr, k)$-APX-EHS for any constant $d$. Taking in particular any $d \geq 10 k / r$ completes the proof.
\end{proof}

It turns out that this problem is somewhat inconvenient for our purposes. Considering again the constraint satisfaction view: an approximate solution to the previous problem satisfies a $1 / \alpha$-fraction of all the constraints. However, each variable may satisfy or fail to satisfy various constraints. We will prefer a formulation of exact hitting set where we \emph{delete} as few variables as possible, as well as all the constraints incident to them, while attempting to reach an instance which is perfectly solvable. This variant concentrates the ``bad'' constraints around a small set of ``bad'' vertices. By analogy to ``almost 3-colouring'' (see, e.g.~\cite{almost3colouring}) we name this problem \emph{Almost Exact Hitting Set}. We now define this problem, in the bipartite graph formulation.

\begin{definition}[$(\beta, r, k)$-\textsc{Almost}-EHS]
    In the \thmemph{$\beta$-almost $(r, k)$-biregular hitting set} problem (abbreviated $(\beta, r, k)$-\textsc{Almost}-EHS), we are given a $(r, k)$-biregular bipartite graph $G = (A \cup B, E)$.
    A prospective solution to the problem is a pair of sets $S, Z \subseteq A$ such that every $b \in B \setminus N(Z)$ has \thmemph{exactly} one neighbour in $S$.
    The value of a solution $(S, Z)$, denoted by $\lambda(S, Z)$, is given by $|Z| / |A|$. We must decide whether there exists a solution with value 0, or not even any solution with value at most $1 / \beta$.
\end{definition}

\begin{prob}
  \probtitle{$(\beta, r, k)$-\textsc{Almost}-EHS}
  \probrow{Input}{$G = (A \cup B, E)$, a $(r, k)$-biregular bipartite graph.}
  
  \probrow{Solution}{Subsets $S, Z \subseteq A$ satisfying, for all $b \in B \setminus N(Z)$, that $| S \cap N(b) | = 1$.
  }
  \probrow{Value}{$\displaystyle\lambda(S, Z) = \frac{|Z|}{|A|}.$}
  \probrow{Yes-inst.}{There exists a solution $(S, Z)$ with $\lambda(S, Z) = 0$.}
  \probrow{No-inst.}{All solutions $(S, Z)$ have $\lambda(S, Z) > 1 / \beta$.}
\end{prob}

It is not difficult to reduce APX-EHS to \textsc{Almost}-EHS~-- indeed, this reduction can be generically applied to any constraint satisfaction problem where every variable is contained within the same number of constraints, and every constraint contains the same number of variables.

\begin{theorem}
    Let $\beta_2 = 2k_2$. Then, $(\beta_2, r_2, k_2)$-\textsc{Almost}-EHS is NP-hard.
\end{theorem}

\begin{proof}
    We reduce from $(2, r_2, k_2)$-APX-EHS.

    \begin{description}
    \item[Reduction.] Our reduction is the ``trivial reduction'': it does not change the graph, taking as input the $(r_2, k_2)$-biregular bipartite graph $G$, and returning $G$. Obviously this is in polynomial time.
    
    \item[Completeness.] Let $G$ be a \textsc{Yes}-instance of $(2, r_2, k_2)$-APX-EHS; that is, there exists $S \subseteq A$ with $\xi(S)=1$. In other words, every vertex in $B$ has exactly one neighbour in $S$. This implies that $(S, \emptyset)$ is a valid solution to $(\beta_2, r_2, k_2)$-\textsc{Almost}-EHS. Noting that $\lambda(S, \emptyset) = 0$ implies completeness.

    \item[Soundness.] Suppose $(S, Z)$ is a solution to $(\beta_2, r_2, k_2)$-\textsc{Almost}-EHS, with value at most $1 / \beta_2$ i.e.~$|Z| \leq |A| / \beta_2$. First, observe that by biregularity, $r_2|A| = k_2 |B|$, hence we can deduce that
    \[
    |N(Z)| \wwrel\leq r_2 |Z| \wwrel\leq \frac{r_2 |A|}{\beta_2}
    \wwrel= 
    \frac{k_2 |B|}{2 k_2}
    \wwrel= \frac{|B|}{2}.
    \]
    Now, by definition, every $b \in B \setminus N(Z)$ is hit by $S$ i.e.~has exactly one neighbour in $S$. Hence, $\xi(S) \geq \frac{1}{|B|} |B \setminus N(Z)| \geq \frac{1}{2}$. This completes the soundness proof.
    \end{description}
\end{proof}

Finally, let us give a gap-version of (U-)MINETREX. The main point here is that this gap version is no harder than approximately solving MINETREX with additive error $\delta n$.

\begin{prob}
  \probtitle{$\delta$-MINETREX}
  \probrow{Input}{$G=(V,E)$, a weakly connected directed graph, and $t \in \mathbb{R}_{\geq 0}$.}
  
  \probrow{Solution}{A spanning tree $T \subseteq E$.
  }
  \probrow{Yes-inst.}{There exists a solution $T$ with $H(G-T) \leq t$.}
  \probrow{No-inst.}{All solutions $T$ have $H(G-T) > t + \delta n$.}
\end{prob}

\begin{theorem}
     $\delta$-MINETREX is $\NP$-hard for some small enough $\delta$. Hence it is \NP-hard to approximate MINETREX with additive error $\delta n$.
\end{theorem}

\begin{proof}
    We reduce from $(\beta_2, r_2, k_2)$-\textsc{Almost}-EHS. For brevity, let $\beta = \beta_2, r = r_2, k = k_2$.

        \paragraph{Reduction} %
        Let $G_{\beta}=(A \cup B, E_\beta)$ be a graph instance of $(\beta, r, k)$-\textsc{Almost}-EHS. Let $n_A = |A|, n_B = |B|$. Construct a digraph $G$ as follows: $G$ starts from $G_\beta$, orienting all edges from $B$ to $A$. Finally, add a binary tree $T_A$ whose leaves are $A$ and otherwise uses new vertices, with edges oriented towards the leaves. Let $n = (2n_A - 1) + n_B$ be the number of vertices in $G$.

        Let $m' = n_B (k - 1)$ and $s(d) = \lg(m' / d)$. It is easy to check that, in any solution $T$ to our $\delta$-MINETREX instance, there are in total $m'$ edges left in $G - T$. Furthermore, if the indegree sequence is given by $d_1 \geq \cdots \geq d_{n}$, we have $H(G - T) = \sum_i d_i s(d_i)$.
        Finally, let $t = m' s(r)$.

        \paragraph{Completeness} Suppose $G_{\beta}=(A \cup B, E_\beta)$ is a \textsc{Yes}-instance of $(\beta, r, k)$-\textsc{Almost}-EHS; that is there exists $S \subseteq A$ such that every vertex in $B$ has exactly one neighbour in $S$. Then, $G$ is a \textsc{Yes}-instance of $\delta$-MINETREX.

        Indeed, construct the spanning tree $T$ as follows: $T$ consists of the edges of the binary tree $T_A$ we added, together with all edges incident to $S$. $T$ is now a spanning tree, since it connects the root of $T_A$ to all of $T_A$ (and hence $A$) by construction of $T_A$, and to all of $B$, since every vertex in $B$ is adjacent to exactly one vertex in $S$. It is a tree since (i) no cycle can be formed within $T_A$, and (ii) all the vertices in $B$ have (undirected) degree 1 in $T$, and hence cannot be part of a cycle.
        
        The sorted indegree sequence of $G-T$ is given by $(r, r, \dots, r, 0, 0,\dots,0)$. It is easy to compute, then, that $H(G - T) = m' s(r)$, as required.

        \paragraph{Soundness}
        Suppose $T$ is a solution to $\delta$-MINETREX with value at most $t+\delta n$; thus, $H(G-T) \leq m' s(r) + \delta n$. As a preliminary step, let us argue that we can always assume that every edge within $T_A$ is included in $T$. We do this in two steps.
        \begin{enumerate}
        \item Consider any edge $(u, v)$ in $T_A$ where $v \not \in A$, and $(u, v)$ \emph{not} in $T$. If this edge is not included in $T$, then the indegree of $v$ in $G - T$ is by construction 1. A simple calculation shows that it never increases $H(G - T)$ to add $(u, v)$ to $T$, and to remove any edge that closes a cycle in $T$. In particular, since $T_A$ is a tree, it is always possible to have $(u, v)$ replace an edge not belonging to $T_A$.
        \item Next consider an edge $(u, v)$ in $T_A$ where $v \in A$, but $(u, v)$ \emph{not} in $T$. Observe that all the neighbours\footnote{We say that $u$ and $v$ are neighbours in some \emph{directed} graph if and only if at least one of $(u, v)$ or $(v, u)$ are edges in that graph.} of $v$ in $T$ are in $B$. Hence, since $T$ is a spanning tree, there exists some edge $(w, v) \in T$ with $w \in B$ such that $(u, v)$, $(w, v)$ and a simple path from $u$ to $w$ form a cycle. Adding $(u, v)$ to $T$ and removing $(w, v)$ from $T$ leaves us with a spanning tree which gives rise to the same sequence of indegrees, thus not changing $H(G - T)$.
        \end{enumerate}
        Hence we see that the edges of $T_A$ can be assumed to belong to $T$.

        Next we claim that for some constant $C$ that depends only on $k, r$, we have that  $G-T$ contains at most $\delta n / C$ vertices $v$ with indegree $0<d_v<r$.
        
        By construction, no vertex can have indegree $> r$ in $G - T$. Thus, the sorted indegree sequence of $G - T$ must have form $r = \cdots = r > d_i \geq \cdots \geq d_{n}$. Let $\ell$ be the number of appearances of $r$ in this sequence; by \wref{lem:domination}, we have that the minimum-entropy indegree sequence that begins with $\ell$ copies of $r$ is given by
        \[
        \underbrace{r = \cdots = r}_{\textstyle\ell} \wrel> \underbrace{r - 1 = \cdots = r - 1}_{\left \lfloor \tfrac{m' - r \ell}{r - 1}\right \rfloor} \wrel> d_0 \wrel\geq 0 = \cdots = 0,
        \]
        where $0 \leq d_0 < r- 1$ is given by the remainder of dividing $m' - r \ell$ by $r - 1$. Recalling that $\lfloor \frac{m' - r\ell}{r - 1}\rfloor = \frac{m' - r\ell - d_0}{r - 1}$, we see that the entropy of this sequence is given by
        \begin{align*}
            &%
            	r \ell s(r) \bin+ (m' - r \ell - d_0) s(r - 1) \bin+ d_0 s(d_0)
        \\	&\wwrel=
	            m' s(r - 1)
	            \bin- r \ell \left(s (r - 1) - s(r)\right)
	            \bin+
	            d_0 (s(d_0) - s(r - 1))
	     \\	&\wwrel=
	            m' s(r)
	            \bin+ m'(s(r - 1) - s(r))
	            \bin- r \ell \left(s (r - 1) - s(r)\right)
	            \bin+
	            d_0 (s(d_0) - s(r - 1))
	     \\	&\wwrel=
            m' s(r)
            \bin+ (m' - r \ell) (s(r - 1) - s(r)) \bin+ d_0 (s(d_0) - s(r - 1))
         \\	&\wwrel\leq 
         	m' s(r) \bin+ \delta n\qquad\qquad\text{(by assumption).}
        \end{align*}
        Hence, letting $C = s(r - 1) - s(r) = \lg(r / (r - 1))$, a positive constant, and $D = d_0 (s(d_0) - s(r - 1)) = d_0 \lg((r - 1) / d_0)$, a non-negative number, we have
        \[
        C (m' - r \ell) + D \wwrel\leq \delta n,
        \]
        or equivalently
        \[
        m' - r \ell \wwrel\leq \frac{\delta n - D}{C} \wwrel\leq \frac{\delta n}{C}.
        \]
        Now, since $G - T$ has $m'$ edges left, and $r \ell$ of those are pointing towards vertices of indegree $r$, this means that at most $m' - r \ell \leq \delta n / C$ of them are pointing towards vertices of indegree less than $r$. In particular, if each of these points towards a different vertex, we see that there are at most $\delta n / C$ vertices in $G - T$ with indegree between $1$ and $r - 1$, proving the claim we began with.

    Next, observe that $r n_A = k n_B$, and furthermore $n = (2n_A - 1) + n_B$. Hence $n = (2 + (r/k))n_A - 1\leq (2 + (r / k))n_A$, and furthermore there are at most $\delta (2 + (r / k))n_A / C$ vertices in $G - T$ with indegree between 1 and $r - 1$. Set $\delta$ small enough so that this number is at most $n_A / \beta$.

    We now claim that setting $S$ to be the set of vertices in $A$ with indegree 0 in $G - T$, and $Z$ to be the set of vertices in $A$ with indegree between 1 and $r - 1$ in $G - T$ is a valid solution to the original $(\beta, r, k)$-\textsc{Almost}-EHS instance, with value $\lambda(S, Z) = |Z| / n_A \leq 1 / \beta$. To see why this is the case, consider any vertex $u \in B$ that does not have a neighbour in $Z$. Considering only the adjacencies in $T$ now, we see that $u$ must have a non-leaf neighbour $v \in A$.\footnote{This is because $u$ must be connected to the root of $T_A$ through the spanning tree $T$.} Note then that $v$ has indegree at least 2 in $T$; hence in $G - T$ it has indegree at most $r - 1$. But by assumption $v \not \in Z$; hence the indegree of $v$ in $G - T$ is 0, and $v \in S$~-- so we have shown that every vertex in $B \setminus N(Z)$ has \emph{at least} one vertex in $A$ as a neighbour.

    To see why any vertex $u \in B \setminus N(Z)$ has \emph{at most} one neighbour in $S$, suppose $u$ has two such neighbours $v, w \in S$. Then this induces a cycle in $T$, since $v$ and $w$ are connected via $T_A$ (as all of $T_A$ is included in $T$). This is not possible since $T$ is a tree.
\end{proof}

    The undirected case is at least equally as hard as the directed one. Our reduction is identical, as is the completeness analysis. The soundness analysis is somewhat more difficult; we sketch the necessary changes:
    \begin{itemize}
        \item It is a bit more difficult to insure that $T_A$ is included in $T$. In particular, our argument relied on the fact that vertices from $T_A$ but not in $A$ can only have indegree 0 or 1 in $G - T$. Now the vertices internal to the tree $T_A$ may have indegree 2. Nevertheless, there are at most $\Theta(\delta n)$ vertices with indegree 1 or 2 in $G - T$ (as $r \geq 10k > 2$). Consider any such vertex $u$, and suppose that it has two edge pointing towards it in $G - T$, namely $(u, v)$ and $(w, v)$. If we were to add $(u, v)$ and $(w, v)$ to $T$, and remove two edges not from $T_A$ to eliminate the created cycles, this can only increase $H(G - T)$ by $\Theta(1)$. This is because the difference in cost for the edges we add and remove themselves is at most $\lg(r) = O(1)$ for each addition/removal pair. The edges we remove from $T$ may also affect the costs of other edges~-- but since this only increases indegrees in $G - T$, it only \emph{lowers} their costs. Hence overall, we may assume that $T_A$ is included in $T$ with at most a penalty of $\Theta(\delta n)$~-- thus the rest of our reduction works, so long as we set $\delta$ even smaller.
        
        \item Since $r \geq 10k$, we see again that the maximum indegree that could be achieved in the graph, after removing a spanning tree, is $r$. In the directed case $B$ had indegree 0 no matter what we did; here it is necessary to observe that the indegrees of $B$ are never $\geq r$, no matter how we orient the graph.
        \item Finally, observe that the graph we used in our reduction has constant maximum (undirected) degree. Hence, $n$ and $m$ are asymptotically equivalent, and we may freely go from a linear error in $n$ to a linear error in $m$.
    \end{itemize}

    Hence we deduce the following.
    
    \begin{corollary}
        $\delta$-U-MINETREX is \NP-hard for all small enough $\delta$~-- even on graphs with constant maximum (undirected) degree. Hence it is \NP-hard to approximate U-MINETREX with additive error $\delta n$, or with additive error $\delta' m$ for some constants $\delta, \delta'$.
    \end{corollary}

\section{Ultrasuccinct Unlabelled Graphs}
\label{sec:data-structures}

In this section, we show how to support efficient navigational queries on our
TREX-compressed graph representation, turning TREX into a \emph{data structure} for unlabelled graphs.
We extend the techniques from~\cite{IsmailiAlaouiNamrataWild2025} from preferential-attachment graphs to
\emph{general unlabelled graphs}, both undirected and directed.

\subsection{Preliminaries}
For the reader's convenience, we recall the results on succinct data structures that we build on in this section.
First, we cite the compressed bit vectors of Pătrașcu~\cite{Patrascu2008}.
(A more practical variant with worse space exists, as well~\cite{RamanRamanRao2007}.)

\begin{lemma}[Compressed bit vector]
\label{lem:compressed-bit-vectors}
	Let $B[1..n]$ be a bit vector of length~$n$, containing $m$ $1$-bits.
	For any constant $c>0$, there is a data structure using
	\(
			\lg \binom{n}{m} \wbin+ O\bigl(\frac{n}{\lg^c n}\bigr)
		\wwrel\le 
			m \lg \bigl(\frac nm\bigr) \wbin+ O\bigl(\frac{n}{\lg^c n}+m\bigr)
	\)
	bits of space that
	supports in $O(1)$ time operations 
	(for $i \in [1..n]$):
	\begin{itemize}
		\item $\accessop(B, i)$: return $B[i]$, the bit at index $i$ in $B$;
		\item $\rankop_\alpha(B, i)$: return the number of bits with
		value $\alpha \in \{0,1\}$ in $B[1..i]$;
		\item $\selop_\alpha(B, i)$: return the index of the $i$th
		bit with value $\alpha \in \{0,1\}$.
	\end{itemize}
\end{lemma}

Using wavelet trees, we can support rank and select queries on
arbitrary static strings / sequences $S\in[\sigma]^n$ while compressing them to zeroth-order empirical entropy 
\[H_0(S) \wwrel= \sum_{c=1}^\sigma |S|_c \lg\left(\frac{n}{|S|_c}\right).\]
Here, $|S|_c$ is the number of occurrences of character $c$ in $S$.

\begin{lemma}[Wavelet tree~\cite{Navarro2014}]
\label{lem:wavelet-tree-simple}
	Let $S[1..n]$ be an array with entries $S[i]\in\Sigma = [1..\sigma]$. There is a data structure using $H_0(S) + o(n)$ bits of space that supports the following queries in $O(\lg \sigma)$ time (without access to $S$ at query time):
    \begin{itemize}

	\item $\accessop(S, i)$: return $S[i]$, the symbol at index $i$ in $S$;
	\item $\rankop_\alpha(S, i)$: return the number of indices with value $\alpha \in \Sigma$ in $S[1..i]$;
	\item $\selop_\alpha(S, i)$: return the index of the $i$th occurrence of value $\alpha \in \Sigma$ in $S$.
    \end{itemize}
\end{lemma}

We lastly need a succinct data structure for ordinal trees with support for levelorder indices (BFS visit order).
Since we only need very basic queries, the classical LOUDS representation~\cite{Jacobson1989}
is sufficient.  
(Further operations can be supported via tree covering~\cite{HeMunroNekrichWildWu2020}.)

\begin{lemma}[{Succinct ordinal trees}]
    \label{lem:tree-succinct}
    Let $T$ be an ordinal tree on $n$ vertices.
    There is a data structure using $2n+o(n)$ bits of space that supports the following queries in $O(1)$ time 
    (where nodes are identified with their levelorder index in $T$):
    \begin{itemize}
        \item $\mathsf{parent}(T,v)$: return the parent of $v$ in $T$;
        \item $\mathsf{degree}(T,v)$: return the number of children of $v$ in $T$;
        \item $\mathsf{child}(T,v,i)$: return the $i$th child of $v$ in $T$;
        \item $\textsf{child-rank}(T, v, c)$: assuming $c$ is a child of $v$, return the index $i$ so that $c = \mathsf{child}(T,v,i)$.
    \end{itemize}
\end{lemma}

\subsection{Adjacency string representation}
\label{sec:adjacency-string}

We begin with a space-efficient variant of the textbook adjacency-list representation for \emph{labelled} graphs,
which we will build on below.

Given a directed graph $G=(V,E)$ with $V=[n]$, we define the \emph{adjacency string} $A = A(G)$ as the concatenation of the outneighbourhoods of all vertices. Formally,  $A(G) = [N^+(v_1), N^+(v_2), \ldots, N^+(v_n)]$, where $N^+(v) = \{ u \in V : (v,u) \in E \}$ denotes the outneighbourhood of vertex $v$; we similarly define
$N^-(v) = \{ u \in V : (u,v) \in E \}$.
Note that $A$ crucially depends on the order/names of vertices.

Note further that $A$ itself does not allow to reconstruct $G$ since it does not encode the boundaries between neighbourhoods resp.\ the outdegrees $d^{\text{out}}_v$ of vertices.
We thus augment $A$ with a bitvector $S[1..n+m] = S(G)$ encoding the unary degree sequence $d^{\text{out}}_1,\ldots,d^{\text{out}}_n$ (or equivalently, the \underbar starting indices of the next neighbourhood) as follows: for $v=1,\ldots,n$ append to $S$ one $1$ followed by $d^{\text{out}}_v$ many $0$s.
So $S$ contains exactly $n$ one-bits.

Using compressed bitvectors (see \wref{lem:compressed-bit-vectors}), $S$ can be stored
with $n \lg(m/n)+O(n)$ bits of space, which is never worse than the $\Theta(n \log n)$ bits of space used by the $n$ pointers in a standard adjacency-list representation, and substantially better if $G$ is sparse.
Moreover, we can find the starting index of any outneighborhood using rank and select (details below).

\subsection{TREX Data Structure}

We will now describe our data structure for compressing unlabelled directed graphs. 

Let $G = (V,E)$ be a directed graph on $n = |V|$ vertices and~$m = |E|$ edges. 
As a first compression step, we compute the tree extraction spanning tree~$T$ in~$G$ using our MST-approximation (\wref{sec:approximation}). The vertices of the spanning tree~$T$ are the same as~$G$, \ie, $V(T) = V$. We will store~$T$ using \wref{lem:tree-succinct}. For this, we can root $T$ arbitrarily and choose an ordering of children arbitrarily. 
Nodes in $T$ are referred to by their levelorder indices. 
Note that in this ordering all children of a node in $T$ form an \emph{interval}
in the levelorder indices.
We do a level order traversal of the spanning tree $T$ and compute the renaming mapping for vertices of $G$ according to their level order indices from~$\{1, \ldots, n\}$. 

We now by represent $G-T$ using the adjacency-string encoding from above (\wref{sec:adjacency-string}),
using the levelorder indices from $T$ to order the vertices.
Let $A' = A(G-T)$ and $S' = S(G-T)$ denote the adjacency string and unary outdegree sequence of $G-T$, respectively.
Compared to $A(G)$, the sequence $A'$ contains exactly $n-1$ fewer vertex symbols (and similarly for $S'$). 
The residual sequence $A'$ is stored using a wavelet tree (\wref{lem:wavelet-tree-simple}) over the alphabet~\(\{1,2,\ldots,n\}\).

Edge orientations are implicitly stored in our representation, as adjacency lists store only outgoing edges. However, we have to explicitly store the direction of tree edges since $T$ is an undirected spanning tree.
We thus store a bitvector $D$, indexed so that $D[v]$ is $1$ iff $v$'s parent edge is pointing towards (\ie{} \underbar down to) $v$. In \wpref{fig:ds-example}, the ultrasuccinct TREX data structure on a directed graph with 8 vertices is illustrated.

For undirected graphs, 
we find the spanning tree \emph{and orientation of the edges of $G-T$}, 
but we do not need to store the direction of tree edges.
Otherwise we follow the same procedure.

\subsection{Operations}
\label{sec:operations-G}

Recall that we store~$G = (V, E)$, where~$|V(G)| = n$ and~$E(G) = m$, using an ordinal tree~$T$, a wavelet tree~$A' = A(G-T)$ and a bitstring~$S'=S(G-T)$ of length~$m-(n-1)+n = m+1$ with~$n$ many 1s to mark the start of the adjacency list of each vertex in~$A'$.

\begin{description}
\item[Outdegree of~$v$.] 
	The outdegree is the sum of the outgoing edges stored in $A'$ and the outgoing edges stored in $T$.
	Note that since $T$ is an unordered spanning tree, any subset of $v$'s child edges and its parent edge may be outgoing edges.
	
	The number of $v$'s outgoing edges in $A'$ is given by $\selop_1(S', v+1) - \selop_1(S', v) - 1$. 
	(Here, we assume that $\selop_1(S', n+1) = |S'|+1 = m+2$.)

    The number of outgoing edges in $T$ is given by the sum of (a) $1$ if $D[v] = 0$ (and $v\ne 1$) and 
    (b) the number of $0$s in $D[c_1 .. c_k]$ (computable via $\rankop_0(D,c_k) - \rankop_0(D,c_1-1)$) where 
    $c_1 = \mathsf{child}(T,v,1)$ and 
    $c_k = \mathsf{child}(T,v,\mathsf{degree}(T,v))$.
	
\item[Indegree of~$v$.] 
    As for the outdegree, we obtain the indegree of a vertex $v$ as the sum of the incoming edges in $A'$ and the incoming edges in $T$.
    The latter is entirely symmetric to outdegree, just counting $1$s in $D$ instead of $0$s.

    The indegree in $A'$ equals the number of occurrences of ``letter'' $v$ in the string $A'$, which is given by $\rankop_v(A',m')$ for $m'=|A'| = m-n+1$.
 	
\item[Deciding if $u$ and $v$ are adjacent.] 
    Vertices $u$ and $v$ are adjacent, if and only if either 
    (1) $u$ is the parent of $v$ in $T$, $\mathsf{parent}(T, v) = u$, 
    (2) vice versa, $\mathsf{parent}(T, u) = v$, 
    (3)~$v$ occurs in $u$'s adjacency list in~$A'$ or 
    (4)~$u$ occurs in $v$'s adjacency list in~$A'$.
    
    Conditions (1) and (2) are directly supported on~$T$; for
	(3) and (4), we can use rank and select on~$S'$ and~$A'$ as follows. For condition (3), we first find the index of~$u$ in~$S'$ using $s = \selop_1(S',u) - u$. This marks the start of the adjacency list of~$u$ (we have to subtract the $u$ one-bits in $S'$). 
    Therefore, if~$v$ occurs in~$A'[s+1 \ldots s+\mathsf{outdegree}(u)-1]$, this implies~$u$ and~$v$ are adjacent. In other words, if~$\rankop_v(A',s + \mathsf{outdegree}(u)-1) - \rankop_1(A', s) \geq 1$, then~$v$ occurs and~$u$ and~$v$ are adjacent. Condition (4) is similar.
	
\item[Computing the~$i$th outneighbour of $v$ ($\mathsf{N_{out}}(v, i)$).] 
    Since adjacency lists in general do not follow any specific order,
    we can define a convenient ordering of neighbours.
    By convention, we will consider all outneighbours of $v$ in $T$ to precede all outneighbours of $v$ in $A'$.
    
    The $i$th outneighbour of $v$ in $G$ can then be determined as follows: 
    If $i \le d_T$ for $d_T$ the outdegree of $v$ in $T$, then
    we find the neighbour in $T$; otherwise, it is stored in the wavelet tree~$A'$.
    $d_T$ is computed as in outdegree.

    Assume first that $i\le d_T$.
    If $i=1$ and $D[v] = 0$, return the parent of $v$ in $T$.
    Otherwise, we have to find the $j$th outgoing child in $T$
    where $j=i$ if $D[v]=1$ and $j=i-1$ otherwise.
    This child is given by $\selop_0(D,j+o)$ for $o=\rankop_0(D,c_1-1)$ with $c_1=\mathsf{child}(T,v,1)$.

    Assume now that $d_T< i \le d^{\text{out}}_v$.
    We have to find the $j$th outgoing edge from $v$ in $A'$ for 
    $j= i-d_T$.
    $v$'s outneighbourhood in $A'$ begins at $s = \selop_1(S',v)-v$, so we find our neighbour at 
    $\accessop(A', s+j)$.

\item[Computing the~$i$th inneighbour of~$v$ ($\mathsf{N_{in}}(v, i)$).] 
    As for outneighbours, we consider inneighbours in $T$ to precede those in $A'$.
    Let (with slight abuse of notation) here $d_T$ denote the \emph{in}degree of $v$ in $T$.
    Then the case $i\le d_T$ is symmetric to $i$th outneighbour.

    It remains to consider $d_T< i \le d^{\text{in}}_v$
    and select the $j$th inneighbour of $v$ in $A'$ for $j=i-d_T$.
    The $j$th occurrence of $v$ in $A'$ is found at index
    $y = \selop_v(A',j)$; that index $y$ belongs to vertex $w = \rankop_1(S',y)$, so the $j$th incoming edge originates from $w$.
\item[Outneighbour rank ($\mathsf{N_{out}\text-rank}(v,w)$).]
    In preparation of the extension with twin removal (\wref{sec:twin-removal}),
    we implement the following operation, which may also be of independent interest:
    $\mathsf{N_{out}\text-rank}(v,w)$, which returns the index $i$ so that
    $w = \mathsf{N_{out}}(v,i)$, assuming that $w$ is an outneighbour of $v$ in $G$.

    We check similar cases as when finding the $i$th outneighbour.
    First, we check if $w$ is $v$'s parent in $T$ and $D[v] = 0$; if so, return $1$.
    If conversely $v$ is $w$'s parent and $D[w] = 1$, we find the number $j$ of out-children of $v$ in $T$ before $w$ via
    $j=\rankop_0(D,w) - \rankop_0(D,c_1-1)$ with $c_1=\mathsf{child}(T,v,1)$.
    The result is then $j-[D[v]=0]$.
    Otherwise $w$ occurs in $v$'s section of $A'$. To find its position, we compute
    $j = \selop_w(A',r+1) - s$ where $r=\rankop_w(A',s-1)$ and $s=\selop_1(S',v)-v$;
    then the answer is $j + d_T$ for $d_T$ the outdegree of $v$ in $T$.
    
\item[Inneighbour rank ($\mathsf{N_{in}\text-rank}(v,w)$).]
    Similar to $\mathsf{N_{out}\text-rank}(v,w)$, we also implement inneighbour rank.
    The checks in $T$ are entirely symmetric to the above, so we only need to consider the case that the edge $(w,v)$ is represented in $A'$,
    \ie, we are looking for the occurrence of $v$ in $w$'s section, and need to determine the number $j$ of incoming edges, \ie, occurrences of $v$, up until this occurrence.
    We compute $j = \rankop_v(A',s-1)+1$ for $s = \selop_1(S',w)-w$ the beginning of $w$'s segment;
    the answer then is $j+d_T$ where $d_T$ denotes the indegree of $v$ in $T$
    (computed as above).    
\end{description}

\subsection{Space and time analysis}

In this section, we prove the space and time complexity of our succinct graph representation, restated here for convenience.

\ultrasuccincttrex*

\begin{proof}
The space for $A'$ is $H_0(A') + o(m) = H(G-T) + o(m)$.
The space for $S'$ is $\lg\binom{m+1}{n} + O(m/\log^c m) = \lg\binom{m+1}{n} + o(m)$.
We have $\lg\binom{m+1}{n} \le \lg\binom{m}{n} + \Oh(\lg m) \le n \lg(m/n) + n/\ln 2 \le n \lg(m/n) + n/\ln 2 + \Oh(\lg m)$.
The data structure for $T$ uses $2n+o(n)$ bits of space,
and $D[1..n]$ contributes another $n+o(n)$ bits for an uncompressed bitvector with rank and select support.

The running time of operations is dominated by the operations on $A'$: \textsf{access}, \textsf{rank}, and \textsf{select} take $O(\lg \sigma)$ time on the wavelet tree, where our $\sigma$ is~$n$ here.
All operations on $S'$, $D$ and $T$ run in constant time.
The outdegree of a vertex does not need operations on $A'$.
\end{proof}

\section{Twin Removal}
\label{sec:twin-removal}

We can further compress certain graphs by first contracting vertices with identical neighbourhood (``twins'').
This is particularly effective on unlabelled graphs, since it basically suffices to remember how many twin copies each vertex has.
The resulting reduced (twin-free) graph is then represented using \wref{thm:ultrasuccinct-trex-directed}

We formally define the notion of~\emph{twins} as follows: Let $G = (V,E)$ be a graph, where~$|V| = N$. For a vertex $v \in V$, let $N(v) = \{ u \in V \setminus \{v\} : \{u,v\} \in E \}$ denote the open neighbourhood of $v$. Two distinct vertices $u,v \in V$ are called \emph{twins} if $N(u) = N(v)$; for directed graphs,
two vertices $u,v \in V$ are \emph{twins} if $N^+(u) = N^+(v)$ and $N^-(u) = N^-(v) $.
We focus again on digraphs.

Being twins defines an equivalence relation on the vertex set $V$.
Let~$ V = C_1 \,\dot\cup\, C_2 \,\dot\cup\, \cdots \,\dot\cup\, C_{n}$ be the resulting partition into \emph{twin classes}, where vertices within the same class are pairwise twins.
(This is also formally stated as Lemma~1 of \cite{TurowskiMagnerSzpankowski2020}).
Contracting all twin classes into single vertices yields the \emph{twin-reduced graph}.
More formally, we define $G_0 = (V_0, E_0)$ as follows:
The vertex set is $V_0 = \{v_1,\ldots,v_n\}$, where each $v_i$ corresponds to the twin class $C_i$. We have an edge $(v_i,v_j) \in E_0$ if and only if $(u,w) \in E$ for some (equivalently, all) $u \in C_i$, $w \in C_j$. By construction, the graph $G_0$ is \emph{twin-free}, 
\ie{} no two vertices in $G_0$ have identical open neighbourhoods.

\subsection{Operations with twin removal}

As we will show, 
we can implement operations on the original graph $G$ using operations on $G_0$ as a black box.
We only need to store the sizes of the twin classes on top.
We use the unary encoding of $[|C_1|, |C_2|,\ldots,|C_n|]$ in a compressed bitvector for that; 
more precisely, for each twin class $C_i$, we write a $1$-bit in $B$ followed by $|C_{i}-1|$ many $0$-bits.

While some operations can be supported equally efficiently, 
an exception is the random access to neighbourhood, \ie, operations
$\mathsf{N_{out/in}}(v,i)$, which the twin-reduced representation can not support efficiently. 
(We would need prefix sum of (out/in)degrees of vertices in $G_0$, not just access to their individual degrees.)
We can however support a slightly weaker operation that allows equally efficient traversals through the neighbourhood:
$\mathsf{N_{out/in}\text-next}(v,w)$ yields the next out/in neighbour of $v$ \emph{after} $w$ resp.\ the first neighbour of $v$ if $w=\texttt{null}$.
Note that this is indeed what a linked list of outneighbours in a textbook adjacency-list representation offers.

\paragraph{Vertex ids}

For a vertex $v \in V(G)$, the vertex in~$G_0$ which is its twin, is called its~\emph{representative} $\mathsf{rep}(v)$. For~$v \in V(G_0)$, the representative of $v$ is $v$ itself.
We represent $G_0$ using our TREX data structure (\wref{thm:ultrasuccinct-trex-directed}).
Thus the tree~$T$ induces a level-order indexing of the vertices of the twin-reduced graph~$G_0$, 
assigning them the labels $\{1,\ldots,n\}$.
Twins are assigned consecutive labels from the range $\{n+1,\ldots,N\}$, in order of their representatives.
This is conveniently done via bitvector $B$, where each twin corresponds to a 0-bit.
For example, if $N = 10$ and $n = 3$, the bitstring $B = \texttt{1000100100}$
indicates that after the three representatives $1$, $2$, $3$, we assign $4$, $5$, $6$ as ids for the 3 twins of $1$; ids $7$ and $8$ for the two twins of $2$, and finally $9$ and $10$ to the two twins of $3$.

\begin{description}
\item[Representative.]
For vertex $v \notin V(G_0)$, we obtain its \emph{representative} as 
\[\mathsf{rep}(v) \wwrel= \rankop_1\bigl(B,(\selop_0(B, v-n))\bigr).\] 
In other words, we take the number of \texttt{0}s in~$B$ from positions $1$ to $v$, given by~$v-n$, and find the index of the $(v-n)$th \texttt{0} in~$B$ using $\selop$, and then count the number of \texttt{1}s up to that index using $\rankop$. 

\item[Adjacency queries.] 
From the definition, twins have the same neighbourhood. Therefore, for any two vertices $v,w \in V(G)$,~$ (v,w) \in E(G)$ if and only if $(\mathsf{rep}(v), \mathsf{rep}(w)) \in E(G_0)$.
Thus, adjacency queries between two vertices on $G$ directly translate to adjacency queries on their representatives in the twin-reduced graph $G_0$.
\item[Twin queries.] Given a vertex~$v \in V(G_0)$, we can find the labels of its twins using a \textsf{select} query on~$B$. We first find the index of~$v$th \texttt{1} in~$B$ using~$ i = \selop_1(B, v)$. The number of \texttt{0}s before that index is given by~$i-v$. This implies that the twins of~$v$ are precisely the vertex ids~$[n+i-v+1 \wbin{..} n+i-v+|C_v| -1]$. 

\item[Next outneighbour of $v$ after $w$.]
We iterate through the outneighbours of $v$ by twin classes, \ie, in the neighbourhood order of $G_0$, but stepping through all twins of a vertex directly after its representative is output.

If $w$ is \emph{not} a last twin copy of its representative $w_0=\mathsf{rep}(w)$, we simply increment $w$ resp.\ find the first twin of $w_0$ in case $w=w_0$. This can be done as in the twin query above.

If $w$ is the last copy of $w_0$, we now need to advance in $G_0$ to the next neighbour $u_0$ of $v_0=\mathsf{rep}(v)$ after $w_0$; we find this next neighbour as
$u_0 = \mathsf{N_{out}}(v_0, r+1)$ for $r = \mathsf{N_{out}\text-rank}(v_0,w_0)$. (Both operations are in $G_0$).

\item[Next inneighbour of $v$ after $w$.]
Using an entirely symmetric procedure as for outneighbours works; only the queries in $G_0$ need to use $\mathsf{N_{in}}$ instead of $\mathsf{N_{out}}$.

\item[$i$th (out/in)neighbour.] 
Unless $G_0$ is much smaller than $G$, \ie, $G$ has very many twins,
using this random access operation on neighbourhoods is very slow and should be avoided.  For completeness, we sketch how to support it, nevertheless.

Assuming that, for~$i' \leq n-1$, we know how to query the~$i'$th outgoing and~$i'$th incoming vertices in~$G_0$, then, for~$i \leq N-1$, we can query the~$i$th outgoing and~$i$th incoming vertices in~$G$ by iterating over the neighbours in $G_0$ until we reach the correct range for $i$.
Note that each neighbour in $G_0$ can correspond to an arbitrary number of twin neighbours which need to be counted for $i$.
The operation is entirely symmetric between out- and inneighbours, so we generically speak of neighbours here.

For a vertex~$v \in V(G)$, let~$u \in V(G_0)$ be its representative. 
We query for the first neighbour of~$u$ in $G_0$, namely~$x_1$.
If~$i > |C_{x_1}|$, none of the twins of~$x_1$ will be the~$i$th neighbour of~$v$. Let~$y_1 = i - |C_{x_1}|$. We again query for the second neighbour of~$u$, namely~$x_2$. If~$y_1 > |C_{x_2}|$, none of the twins of~$x_2$ will be the $i$th neighbour of~$v$. We continue this process until, for some neighbour~$x_j$ of~$u$, we get~$|C_{x_j}| > y_{j'}$. The~$y_{j'}$th twin of~$x_j$ will be the~$i$th neighbour of the vertex~$v$. We can use the twin query to get the label of the~$i$th vertex.

\item[Degree queries.] In a similar spirit to the $i$th out/in neighbour query as above, the degree query for a vertex~$v$ is not efficiently supported. This is because we need to sum the number of twins over the neighbourhood of~$\mathsf{rep}(v)$ in~$G_0$. We sketch the details for completeness.

Let $v \in V(G)$ and let $i=\mathsf{rep}(v)$. The outdegree and indegree of $v$ in $G$ can be expressed as~$|N^+(v)|= \sum_{(i,j)\in E(G_0)} |C_j|$ and $
|N^-(v)|
=
\sum_{(j,i)\in E(G_0)} |C_j|
$, respectively.
In particular, computing degrees in $G$ reduces to counting the neighbours of the corresponding representative in $G_0$, weighted by the sizes of the associated twin classes.

Specifically, for the outdegree of a vertex~$v$, find all the outneighbours of $i=\mathsf{rep}(v)$ in $G_0$; in other words, if $N^+(i) = \{\, j \mid (i, j) \in E(G_0) \,\}$, then for each $j \in N^{+}(i)$, the number of vertices belonging to the $j$th class is determined using \texttt{$\selop$} queries on the bitvector~$B$ as follows:
        \[
        |C_{j}| \wwrel=
        \begin{cases}
            \ \selop_1(B,j+1)-\selop_1(B,j) -1 & j < n,\\[6pt]
            \,|B| - \selop_1(B,j) & j = n.
        \end{cases}
        \]
In a similar way, we can calculate the indegree and total degree of a vertex in~$G$.

\end{description}

\subsection{Space and time analysis}  
In this section, we prove the space and time complexity of our twin-removal data structure.

\begin{theorem}[Twin removed TREX]
\label{thm:space-time-analysis-G}
For a directed graph~$G$, 
there exists a data structure for the twin-removed TREX representation~-- consisting of~$B$, $T$, $A'$ and~$S$~-- using 
$H_0(A') + 3n + o(m+N) + \lg \binom{m}{n} + \lg \binom{N}{n}$ bits of space, while allowing for operations $\mathsf{adjacency}(u, v)$ to run in $O(\lg n)$ time and iterating through $\mathsf{N_{out}}(v)$ and $\mathsf{N_{in}}(v)$, in $O(\log n)$ time per neighbour. The operations $\mathsf{N_{out}}(v, i), \mathsf{N_{in}}(v, i)$, and $\mathsf{degree}(v)$ run in $O(d \lg n) = O(n\lg n)$ time,
where $d$ is the number of neighbors of $\mathsf{rep}(v)$ in $G_0$.
\end{theorem}

\begin{proof}
From \wref{thm:ultrasuccinct-trex-directed}, we already know the space required by~$T, A'$ and $S$.
Recall that bitvector~$B$ encodes the multiplicities of twins in~$G$.
Using compressed bitvectors (\wref{lem:compressed-bit-vectors}), $B$ with~$n$ many 1s and $|B|=N$ can be stored in~$
\lg \binom{N}{n} + o(N) \le n \lg \frac{N}{n} + O\Bigl(\frac{N}{\lg^c N}\Bigr)$ bits, for any constant $c>0$, while supporting $\rankop$ and $\selop$ in $O(1)$ time.  This proves the space analysis.

The query times for $\mathsf{adjacency}(u, v)$, $\mathsf{N_{out}}\textsf{-next}(v,w)$ and $\mathsf{N_{in}}\textsf{-next}(v,w)$ directly follow from \wref{thm:ultrasuccinct-trex-directed}. 
On the other hand, for operations $\mathsf{N_{out}}(v, i), \mathsf{N_{in}}(v, i)$, and $\mathsf{degree}(v)$, we have to iterate sequentially over the neighbourhoods in $G_0$, with a $O(\lg n)$ time query for each of the $d$ representative neighbours. Clearly, $d\le n$.
\end{proof}

\section{Good graphs classes}\label{sec:good-graphs}

In this chapter, we prove that removing trees often leads to a significant
decrease in the number of bits needed to store a graph.
First, it will be useful to use a variant of \wref{thm:generic_apx}
that is tailored for our purposes in this chapter.  For the following we will often think of
equations holding only up to some error term. For this purpose, when we write for example $x = y \pm \epsilon$,
we mean that $|x - y| \leq \epsilon$.

\begin{theorem}
    Consider the same setup as in \wref{thm:generic_apx}. Then
    \[
        H(d_1 - x_1, \ldots, d_n - x_n) 
        \wwrel=
        H(d_1, \ldots, d_n) \bin+ \sum_i x_i \lg \frac{d_i}{m} \wbin\pm \frac{2k}{\ln 2}.
    \]
\end{theorem}
\begin{proof}
    Indeed, noting that $\optent(x_1, \ldots, x_n) = H(d_1 - x_1, \ldots, d_n - x_n)$, we have, by
    the main inequality of \wref{thm:generic_apx},
    \begin{align*}
        H(d_1 - x_1, \ldots, d_n - x_n) 
    	&\wwrel=
        \optent(x_1, \ldots, x_n) 
    \\	&\wwrel=
        (m - k) \lg(m - k) - \sum_i d_i \lg d_i + \optlin(x_1, \ldots, x_n) \bin\pm \frac{k}{\ln 2}
    \\	&\wwrel= 
        (m - k) \lg(m - k) - \sum_i d_i \lg d_i + \sum_i x_i \lg d_i \bin\pm \frac{k}{\ln 2}
    \end{align*}

    Now, apply \wref{lem:simpleBound} to $(m - k)(\lg m - \lg(m - k))$ to find that $(m - k) \lg (m - k)
    = (m - k) \lg m \pm \frac{k}{\ln 2}$, hence
    \begin{align*}
        H(d_1 - x_1, \ldots, d_n - x_n)
        &\wwrel= 
        (m - k) \lg m - \sum_i d_i \lg d_i + \sum_i x_i \lg d_i \wbin\pm \frac{2k}{\ln 2}
    \\  &\wwrel= 
        \sum_i (d_i - x_i) \lg m - \sum_i d_i \lg d_i + \sum_i x_i \lg d_i \wbin\pm \frac{2k}{\ln 2}
    \\  &\wwrel= 
    	\left(- \sum_i d_i \lg \frac{d_i}{m}\right) + \left(\sum_i x_i \lg \frac{d_i}{m}\right) \wbin\pm \frac{2k}{\ln 2}
    \\  &\wwrel= 
    	H(d_1, \ldots, d_n) + \sum_i x_i \lg \frac{d_i}{m} \wbin\pm \frac{2k}{\ln 2}.
        \iflipics{\qedhere}{}
    \end{align*}
\end{proof}

Let us now apply this theorem to our particular case.

\begin{corollary}\label{cor:pre-LP}
    Consider any digraph $G = (V, E)$ with $n$ vertices and $m$ edges. Consider any spanning tree $T$
of $G$. Assign edge $(i, j)$ weight $w(i, j) = - \lg d_j / m$, where $d_j$ is the indegree of $j$ in $G$.
We have
\[
    H(G - T) \wwrel= H(G) \bin- \sum_{\mathclap{(i, j) \in T}} w(i, j) \wbin\pm \frac{2n}{\ln 2}.
\]
Hence the optimal savings we can gain (up to an error of $\frac{2n}{\ln 2}$) by removing some tree is given by the maximum
cost spanning tree under the weights described above.
\end{corollary}

This is both stronger and weaker than \wref{thm:generic_apx}~-- whereas that theorem told us that linear optimisation
is sufficient for approximate entropy minimisation, this theorem also relates the precise value of
the entropy associated with deleting a tree with a linear function of the edges remaining in the
graph (but with a larger error).

Recall the classic result of~\cite{EdmondsPolytope}: the 0--1 vectors corresponding to spanning trees
of a given graph $G$ form the vertices of a polytope with a simple formulation as an LP. More precisely, we can find a maximum cost
spanning tree in a graph $G = (V, E)$ with indegree sequence $d_1, \ldots, d_n$, with weights as above,
by solving the following linear program:

\begin{maxi*}{x_{e}, e \in E}{-\sum_{(i, j) \in E} x_{ij} \lg \frac{d_j}{m}}{}{}
    \addConstraint{\sum_{e \in E} x_e}{\wrel= n - 1}{}
    \addConstraint{\sum_{e \in E \cap S^2} x_e}{\wrel\leq |S| - 1}{\quad\forall S \subseteq V, 1 < |S| < n}
    \addConstraint{x_e}{\wrel\geq 0}{\quad\forall e \in E}
\end{maxi*}

It is easy to see that eliminating the first equality constraint leads to a polytope whose vertices
are precisely the acyclic subgraphs of $G$. Since we have only non-negative costs and will only consider connected graphs, we may as well
weaken that constraint. This leads to the following equivalent formulation:

\begin{maxi}{x_{e}, e \in E}{-\sum_{(i, j) \in E} x_{ij} \lg \frac{d_j}{m}}{}{}\label{eq:LP}
    \addConstraint{\sum_{e \in E \cap S^2} x_e}{\wrel\leq |S| - 1}{\quad\forall S \subseteq V, |S| \geq 2}
    \addConstraint{x_e}{\wrel \geq 0}{\quad\forall e \in E}
\end{maxi}

(This can also be seen as an equivalent formulation for the independence polytope of the graphical matroid.)
We hence get the following fact, by combining this LP formulation for maximum cost spanning tree,
together with \wref{cor:pre-LP}.

\begin{corollary}\label{cor:LP}
    Consider some connected digraph $G = (V, E)$ with $n$ vertices and $m$ edges, with indegreee sequence $d_1, \ldots, d_n$.
Suppose there exists some values $x_e \in \mathbb{R}_{\geq 0}$ for $e \in E$, such that for all $S \subseteq V$
with $|S| \geq 2$ we have
\[
    \sum_{e \in E \cap S^2} x_e \wwrel\leq |S| - 1.
\]
Then there exists some spanning tree $T$ such that
\[
    H(G - T) \wwrel\leq H(G) + \left(\sum_{(i, j) \in E} x_{ij} \lg \frac{d_j}{m}\right) \bin+ \frac{2n}{\ln 2}.
\]
\end{corollary}

This is precisely what we will use to prove the combinatorial facts we want. Recall that $\delta(G)$ is the ratio of the number of edges of $G$ divided by the number of vertices of $G$, $\delta(G) = |E(G)|/|V(G)|$.

\thmEquanimity*
\begin{proof}
    Consider the values $x_e = n / 2m \alpha \geq 0$. Let us check it follows the conditions
laid out in \wref{cor:LP}. Consider any subset $S \subseteq V$ with $|S| \geq 2$. Suppose $S$ has
$n'$ vertices within it, and $m'$ edges within it. By assumption, $m' / n' \leq \alpha m / n$. So
\[
    \sum_{e \in E \cap S^2} x_e
    \wwrel= \sum_{e \in E \cap S^2} \frac{n}{2m \alpha}
    \wwrel= \frac{m'}{2} \left(\frac{\alpha m}{n}\right)^{-1}
    \wwrel\leq \frac{m'}{2} \frac{n'}{m'}
    \wwrel= \frac{n'}{2} 
    \wwrel\leq n' - 1,
\]
as required. Now let us ensure that the savings is sufficient. Note that (excepting the error of $2n / \ln 2$),
we save
\begin{align*}
    -\sum_{(i, j) \in E} x_{ij} \lg \frac{d_j}{m}
    &\wwrel=
    -\sum_{(i, j) \in E} \frac{n}{2m \alpha} \lg \frac{d_j}{m} 
    \\&\wwrel=
    -\sum_{j \in V} d_j \frac{n}{2m \alpha} \lg \frac{d_j}{m}
    \\&\wwrel=
    \frac{n}{2\alpha} \left(-\sum_{j \in V} \frac{d_j}{m} \lg \frac{d_j}{m}\right)
    \\&\wwrel= \frac{n}{2\alpha} \DegreeEntropy(G),
\end{align*}
as required.
\end{proof}

Finally, by similar techniques we can show that graphs where many vertices have nonzero indegree
must have a high savings.

\thmManyNonzero*

\begin{proof}
    Let $d_1, \ldots, d_n$ be the indegree sequence. Suppose without loss of generality that $d_1 \geq \cdots \geq d_n$.
Suppose that $d_1, \ldots, d_k > 0$ and $d_{k + 1} = \cdots = d_n = 0$. Hence, by assumption $k \geq n / \alpha$.
Consider the sequence $x_{ij} = 1 / 2 d_j$. (Note that for every edge $(i, j) \in E$, $d_j \geq 1$ so this
division is well-defined.) As before we check that this satisfies the conditions of \wref{cor:LP}.
Consider any subset $S \subseteq V$ with at least 2 vertices. Note:
\[
    \sum_{e \in E \cap S^2} x_e
    \wwrel\leq \sum_{\substack{j \in S \\ d_j > 0}} \sum_{(i, j) \in E} x_{ij} 
    \wwrel= \sum_{\substack{j \in S \\ d_j > 0}} \sum_{(i, j) \in E} \frac{1}{2 d_j} 
    \wwrel= \sum_{\substack{j \in S \\ d_j > 0}} d_j \frac{1}{2 d_j}
    \wwrel\leq \frac{|S|}{2} 
    \wwrel\leq |S| - 1.
\]
So again, we satisfy the necessary conditions. Let us now compute how much we save.
\begin{equation}\label{eq:interesting}
\begin{aligned}[b]
    -\sum_{(i, j) \in E} x_{ij} \lg \frac{d_j}{m}
    &\wwrel= 
    	-\sum_{\substack{j \in V \\ d_j > 0}} \sum_{(i, j) \in E} x_{ij} \lg \frac{d_j}{m} 
    \wwrel= -\sum_{\substack{j \in V \\ d_j > 0}} \sum_{(i, j) \in E} \frac{1}{2 d_j} \lg \frac{d_j}{m}
    \\&\wwrel= -\sum_{\substack{j \in V \\ d_j > 0}} d_j \frac{1}{2 d_j} \lg \frac{d_j}{m} 
    \wwrel= -\sum_{\substack{j \in V \\ d_j > 0}} \frac{1}{2} \lg \frac{d_j}{m} 
    \\&\wwrel= \frac{k}{2} \left(-\sum_{\substack{j \in V \\ d_j > 0}} \frac{1}{k} \lg \frac{d_j}{m} \right) 
    \wwrel\geq \frac{n}{2\alpha} \left(-\sum_{\substack{j \in V \\ d_j > 0}} \frac{1}{k} \lg \frac{d_j}{m} \right).
\end{aligned}
\end{equation}
Now, compare $-\sum_{{j \in V : d_j > 0}} \frac{1}{k} \lg \frac{d_j}{m}$ with
\[
    \DegreeEntropy(G) \wwrel= - \sum_{\substack{j \in V \\ d_j > 0}} \frac{d_j}{m} \lg \frac{d_j}{m}.
\]
(We often omit the $d_j > 0$ since we multiply by $d_j$, but here we include it to emphasise we are
summing over the same indices.) We claim that
\begin{equation}\label{eq:theThing}
    -\sum_{\substack{j \in V \\ d_j > 0}} \frac{1}{k} \lg \frac{d_j}{m}
    \wwrel\geq
    -\sum_{\substack{j \in V \\ d_j > 0}} \frac{d_j}{m} \lg \frac{d_j}{m}.
\end{equation}
Substituting into~\eqref{eq:interesting} implies that we save at least $\frac{n}{2\alpha} \DegreeEntropy(G)$,
ignoring the $2n / \ln 2$ error term.

Note that the two terms in~\eqref{eq:theThing} are both weighted averages of $-\lg (d_j / m)$
for $j = 1, \ldots, k$.
Suppose that
\[
    d_1 \geq \cdots \geq d_t \wrel\geq \frac{m}{k} \wrel\geq d_{t + 1} \geq \cdots d_k.
\]
Note that the terms $- \lg (d_1 / m), \ldots, - \lg (d_t / m)$ are weighted more heavily in the right side of~\eqref{eq:theThing},
whereas the terms $- \lg (d_{t + 1} / m), \ldots, - \lg (d_k / m)$ are weighted more heavily in the left side of~\eqref{eq:theThing}.
However it is plain to see that \emph{all} of the former terms are no larger than \emph{all}
of the latter terms. This implies the inequality we want, and the entire result.
\end{proof}

	\myacknowledgements

%
%
\bibliography{references.bib,bib.bib}

@inproceedings{GuruswamiTrevisan2005,
  author       = {Venkatesan Guruswami and
                  Luca Trevisan},
  editor       = {Chandra Chekuri and
                  Klaus Jansen and
                  Jos{\'{e}} D. P. Rolim and
                  Luca Trevisan},
  title        = {The Complexity of Making Unique Choices: Approximating 1-in- k {SAT}},
  booktitle    = {Approximation, Randomization and Combinatorial Optimization, Algorithms
                  and Techniques, 8th International Workshop on Approximation Algorithms
                  for Combinatorial Optimization Problems, {APPROX} 2005 and 9th InternationalWorkshop
                  on Randomization and Computation, {RANDOM} 2005, Berkeley, CA, USA,
                  August 22-24, 2005, Proceedings},
  series       = {Lecture Notes in Computer Science},
  volume       = {3624},
  pages        = {99--110},
  publisher    = {Springer},
  year         = {2005},
  url          = {https://doi.org/10.1007/11538462\_9},
  doi          = {10.1007/11538462\_9},
  bibsource    = {dblp computer science bibliography, https://dblp.org}
}

@article{feige98setcover,
author = {Feige, Uriel},
title = {A threshold of ln n for approximating set cover},
year = {1998},
issue_date = {July 1998},
publisher = {Association for Computing Machinery},
address = {New York, NY, USA},
volume = {45},
number = {4},
issn = {0004-5411},
doi = {10.1145/285055.285059},
journal = {J. ACM},
month = jul,
pages = {634–652},
numpages = {19}
}

@preamble{"\providecommand{\NOOPSORT}[1]{}"}

@book{Mahmoud2008,
	title = {P{\'{o}}lya Urn Models},
	author = {Mahmoud, Hosam M.},
	isbn = {978-1-4200-5983-0},
	publisher = {Chapman {\&} Hall},
	year = {2008}
}

@article{RamanRamanRao2007,
	Author = {Rajeev Raman and Venkatesh Raman and Rao Satti, Srinivasa},
	Doi = {10.1145/1290672.1290680},
	Journal = {{ACM} Transactions on Algorithms},
	Month = {nov},
	Number = {4},
	Pages = {43--es},
	Publisher = {Association for Computing Machinery ({ACM})},
	Title = {Succinct indexable dictionaries with applications to encoding k-ary trees, prefix sums and multisets},
	Volume = {3},
	Year = {2007}
}

@article{CastelliAleardiSchaeffer2008,
  title = {Succinct representations of planar maps},
  volume = {408},
  ISSN = {0304-3975},
  DOI = {10.1016/j.tcs.2008.08.016},
  number = {2–3},
  journal = {Theoretical Computer Science},
  publisher = {Elsevier BV},
  author = {Castelli Aleardi,  L. and Devillers,  O. and Schaeffer,  G.},
  year = {2008},
  month = nov,
  pages = {174–187}
}

@inproceedings{MunroWu2018,
	author    = {J. Ian Munro and
	Kaiyu Wu},
	title     = {Succinct Data Structures for Chordal Graphs},
	booktitle = {29th International Symposium on Algorithms and Computation, {ISAAC}
	2018, December 16-19, 2018, Jiaoxi, Yilan, Taiwan},
	pages     = {67:1--67:12},
	year      = {2018},
	doi       = {10.4230/LIPIcs.ISAAC.2018.67},
	bibsource = {dblp computer science bibliography, https://dblp.org}
}

@article{Kamali2017,
  title = {Compact Representation of Graphs of Small Clique-Width},
  volume = {80},
  ISSN = {1432-0541},
  DOI = {10.1007/s00453-017-0365-6},
  number = {7},
  journal = {Algorithmica},
  publisher = {Springer Science and Business Media LLC},
  author = {Kamali,  Shahin},
  year = {2017},
  month = sep,
  pages = {2106–2131}
}

@inproceedings{KammerMeintrup2023,
  doi = {10.4230/LIPICS.ISAAC.2023.44},
  author = {Kammer,  Frank and Meintrup,  Johannes},
  language = {en},
  title = {Succinct Planar Encoding with Minor Operations},
  journal = {LIPIcs,  Volume 283,  ISAAC 2023},
  volume = {283},
  pages = {44:1--44:18},
  publisher = {Schloss Dagstuhl – Leibniz-Zentrum f\"{u}r Informatik},
  year = {2023},
  copyright = {Creative Commons Attribution 4.0 International license}
}

@inproceedings{BlellochFarzan2010,
  title = {Succinct Representations of Separable Graphs},
  DOI = {10.1007/978-3-642-13509-5_13},
  booktitle = {Combinatorial Pattern Matching (CPM)},
  publisher = {Springer Berlin Heidelberg},
  author = {Blelloch,  Guy E. and Farzan,  Arash},
  year = {2010},
  pages = {138–150}
}

@article{FerresFuentesSepulvedaGagieHeNavarro2020,
  title = {Fast and compact planar embeddings},
  volume = {89},
  ISSN = {0925-7721},
  DOI = {10.1016/j.comgeo.2020.101630},
  journal = {Computational Geometry},
  publisher = {Elsevier BV},
  author = {Ferres,  Leo and Fuentes-Sepúlveda,  José and Gagie,  Travis and He,  Meng and Navarro,  Gonzalo},
  year = {2020},
  month = aug,
  pages = {101630}
}

@article{TsakalidisWildZamaraev2023,
	title = {Succinct Permutation Graphs},
	author = {Konstantinos Tsakalidis and Sebastian Wild and Viktor Zamaraev},
	year = {2023},
	doi = {10.1007/s00453-022-01039-2},
	url = {https://www.wild-inter.net/publications/tsakalidis-wild-zamaraev-2023},
	volume = {85},
	number = {2},
	pages = {509--543},
	eprint={2010.04108},
	archivePrefix={arXiv},
	primaryClass={cs.DS},
	journal = {Algorithmica}
}

@article{CardinalFioriniJoret2008,
  title = {Minimum entropy orientations},
  volume = {36},
  ISSN = {0167-6377},
  DOI = {10.1016/j.orl.2008.06.010},
  number = {6},
  journal = {Operations Research Letters},
  publisher = {Elsevier BV},
  author = {Cardinal,  Jean and Fiorini,  Samuel and Joret,  Gwenaël},
  year = {2008},
  month = nov,
  pages = {680–683}
}

@article{Navarro2014,
  title = {Wavelet trees for all},
  volume = {25},
  ISSN = {1570-8667},
  DOI = {10.1016/j.jda.2013.07.004},
  journal = {Journal of Discrete Algorithms},
  publisher = {Elsevier BV},
  author = {Navarro,  Gonzalo},
  year = {2014},
  month = mar,
  pages = {2–20}
}

@inproceedings{GrossiGuptaVitter2003,
	author = {Grossi, Roberto and Gupta, Ankur and Vitter, Jeffrey Scott},
	title = {High-order entropy-compressed text indexes},
	year = {2003},
	publisher = {SIAM},
	address = {USA},
	booktitle = {Proceedings of the Fourteenth Annual ACM-SIAM Symposium on Discrete Algorithms},
	pages = {841–850},
	numpages = {10},
	location = {Baltimore, Maryland},
	series = {SODA}
}

@article{Gonzalez1985,
  title = {Clustering to minimize the maximum intercluster distance},
  volume = {38},
  ISSN = {0304-3975},
  url = {http://dx.doi.org/10.1016/0304-3975(85)90224-5},
  DOI = {10.1016/0304-3975(85)90224-5},
  journal = {Theoretical Computer Science},
  publisher = {Elsevier BV},
  author = {Gonzalez,  Teofilo F.},
  year = {1985},
  pages = {293–306}
}

@article{MiloShenOrrItzkovitzKashtanChklovskiiAlon2002,
  title = {Network Motifs: Simple Building Blocks of Complex Networks},
  volume = {298},
  ISSN = {1095-9203},
  url = {http://dx.doi.org/10.1126/science.298.5594.824},
  DOI = {10.1126/science.298.5594.824},
  number = {5594},
  journal = {Science},
  publisher = {American Association for the Advancement of Science (AAAS)},
  author = {Milo,  R. and Shen-Orr,  S. and Itzkovitz,  S. and Kashtan,  N. and Chklovskii,  D. and Alon,  U.},
  year = {2002},
  month = oct,
  pages = {824–827}
}

@article{SzpankowskiGrama2018,
  title = {Frontiers of Science of Information: Shannon Meets Turing},
  volume = {51},
  ISSN = {1558-0814},
  DOI = {10.1109/mc.2018.1151004},
  number = {1},
  journal = {Computer},
  publisher = {Institute of Electrical and Electronics Engineers (IEEE)},
  author = {Szpankowski,  Wojciech and Grama,  Ananth},
  year = {2018},
  month = jan,
  pages = {28–38}
}

@inproceedings{IsmailiAlaouiNamrataWild2025,
	title = {Succinct Preferential-Attachment Graphs},
	author = {Ziad {Ismaili Alaoui} and Namrata and Sebastian Wild},
	booktitle = {International Workshop on Graph-Theoretic Concepts in Computer Science (WG)},
	year = {2025},
	publisher = {Springer},
	archivePrefix = {arXiv},
	arxivId = {2506.21436},
	eprint  = {2506.21436},
	primaryClass = {cs.DS},
	url = {https://www.wild-inter.net/publications/ismaili-alaoui-namrata-wild-2025}
}

@article{PatonHartleStepanyantsVanDerHoorn2022,
  title = {Entropy of labeled versus unlabeled networks},
  volume = {106},
  ISSN = {2470-0053},
  DOI = {10.1103/physreve.106.054308},
  number = {5},
  journal = {Physical Review E},
  publisher = {American Physical Society (APS)},
  author = {Paton,  Jeremy and Hartle,  Harrison and Stepanyants,  Huck and van der Hoorn,  Pim and Krioukov,  Dmitri},
  year = {2022},
  month = nov 
}

@article{BenoitDemaineMunroRamanRamanRao2005,
  author    = {David Benoit and Erik D. Demaine and J. Ian Munro and Rajeev Raman and Venkatesh Raman and S. Rao, Srinivasa},
  title     = {Representing Trees of Higher Degree.},
  journal   = {Algorithmica},
  volume    = {43},
  number    = {4},
  year      = {2005},
  pages     = {275-292},
  doi       = {10.1007/s00453-004-1146-6},
}

@article{NavarroSadakane2014,
  doi = {10.1145/2601073},
  year  = {2014},
  month = {may},
  publisher = {Association for Computing Machinery ({ACM})},
  volume = {10},
  number = {3},
  pages = {1--39},
  author = {Gonzalo Navarro and Kunihiko Sadakane},
  title = {Fully Functional Static and Dynamic Succinct Trees},
  journal = {{ACM} Transactions on Algorithms}
}

@article{FarzanMunro2014,
  title = {A Uniform Paradigm to Succinctly Encode Various Families of Trees},
  doi = {10.1007/s00453-012-9664-0},
  year  = {2014},
  month = jun,
  publisher = {Springer},
  volume = {68},
  number = {1},
  pages = {16--40},
  author = {Arash Farzan and J. Ian Munro},
  journal = {Algorithmica}
}

@article{MunroRaman2001,
  title = {Succinct Representation of Balanced Parentheses and Static Trees},
  doi = {10.1137/s0097539799364092},
  year  = {2001},
  month = jan,
  publisher = {SIAM},
  volume = {31},
  number = {3},
  pages = {762--776},
  author = {J. Ian Munro and Venkatesh Raman},
  journal = {{SIAM} Journal on Computing}
}

@InProceedings{Jacobson1989,
  author = 	 {Guy Jacobson},
  title = 	 {Space-efficient static trees and graphs},
  booktitle = 	 {Proceedings of the 30th Annual IEEE Symposium on Foundations of Computer Science},
  year = 	 {1989},
  pages =        {549-554},
  doi       = {10.1109/SFCS.1989.63533}
}

@inproceedings{BoldiVigna2004,
  series = {WWW04},
  title = {The webgraph framework I: compression techniques},
  DOI = {10.1145/988672.988752},
  booktitle = {Proceedings of the 13th international conference on World Wide Web},
  publisher = {ACM},
  author = {Boldi,  P. and Vigna,  S.},
  year = {2004},
  month = may,
  pages = {595–602},
  collection = {WWW04}
}

@article{ChoiSzpankowski2012,
  title = {Compression of Graphical Structures: Fundamental Limits,  Algorithms,  and Experiments},
  volume = {58},
  ISSN = {1557-9654},
  DOI = {10.1109/tit.2011.2173710},
  number = {2},
  journal = {IEEE Transactions on Information Theory},
  publisher = {Institute of Electrical and Electronics Engineers (IEEE)},
  author = {Choi,  Yongwook and Szpankowski,  Wojciech},
  year = {2012},
  month = feb,
  pages = {620–638}
}

@article{TurowskiMagnerSzpankowski2020,
  title = {Compression of Dynamic Graphs Generated by a Duplication Model},
  volume = {82},
  ISSN = {1432-0541},
  DOI = {10.1007/s00453-020-00699-2},
  number = {9},
  journal = {Algorithmica},
  publisher = {Springer Science and Business Media LLC},
  author = {Turowski,  Krzysztof and Magner,  Abram and Szpankowski,  Wojciech},
  year = {2020},
  month = apr,
  pages = {2687–2707}
}

@article{JanssonSadakaneSung2012,
  title = {Ultra-succinct representation of ordered trees with applications},
  volume = {78},
  DOI = {10.1016/j.jcss.2011.09.002},
  number = {2},
  journal = {Journal of Computer and System Sciences},
  publisher = {Elsevier BV},
  author = {Jansson,  Jesper and Sadakane,  Kunihiko and Sung,  Wing-Kin},
  year = {2012},
  month = mar,
  pages = {619–631}
}

@article{FischerPeters2016,
  title = {GLOUDS: Representing tree-like graphs},
  volume = {36},
  ISSN = {1570-8667},
  DOI = {10.1016/j.jda.2015.10.004},
  journal = {Journal of Discrete Algorithms},
  publisher = {Elsevier BV},
  author = {Fischer,  Johannes and Peters,  Daniel},
  year = {2016},
  month = jan,
  pages = {39–49}
}

@article{LuczakMagnerSzpankowski2019,
  title = {Asymmetry and structural information in preferential attachment graphs},
  volume = {55},
  ISSN = {1098-2418},
  DOI = {10.1002/rsa.20842},
  number = {3},
  journal = {Random Structures \& Algorithms},
  publisher = {Wiley},
  author = {Luczak,  Tomasz and Magner,  Abram and Szpankowski,  Wojciech},
  year = {2019},
  month = mar,
  pages = {696–718}
}

@article{AcanChakrabortyJoRao2020,
  title = {Succinct Encodings for Families of Interval Graphs},
  volume = {83},
  ISSN = {1432-0541},
  DOI = {10.1007/s00453-020-00710-w},
  number = {3},
  journal = {Algorithmica},
  publisher = {Springer Science and Business Media LLC},
  author = {Acan,  H\"{u}seyin and Chakraborty,  Sankardeep and Jo,  Seungbum and Rao Satti,  Srinivasa },
  year = {2020},
  month = apr,
  pages = {776–794}
}

@article{AcanChakrabortyJoNakashimaSadakaneRao2022,
  title = {Succinct navigational oracles for families of intersection graphs on a circle},
  volume = {928},
  ISSN = {0304-3975},
  DOI = {10.1016/j.tcs.2022.06.022},
  journal = {Theoretical Computer Science},
  publisher = {Elsevier BV},
  author = {Acan,  H\"{u}seyin and Chakraborty,  Sankardeep and Jo,  Seungbum and Nakashima,  Kei and Sadakane,  Kunihiko and Rao Satti,  Srinivasa},
  year = {2022},
  month = sep,
  pages = {151–166}
}

@article{ChakrabortyJoSadakaneRao2024,
  title = {Succinct data structures for bounded clique-width graphs},
  volume = {352},
  ISSN = {0166-218X},
  DOI = {10.1016/j.dam.2024.03.016},
  journal = {Discrete Applied Mathematics},
  publisher = {Elsevier BV},
  author = {Chakraborty,  Sankardeep and Jo,  Seungbum and Sadakane,  Kunihiko and Rao Satti, Srinivasa},
  year = {2024},
  month = jul,
  pages = {55–68}
}

@misc{BestaHoefler2019,
      title={Survey and Taxonomy of Lossless Graph Compression and Space-Efficient Graph Representations}, 
      author={Maciej Besta and Torsten Hoefler},
      year={2019},
      eprint={1806.01799},
      archivePrefix={arXiv},
      primaryClass={cs.DS}
}

@inproceedings{MunroNicholsonSeelbachBenknerWild2021,
	title = {Hypersuccinct Trees -- New Universal Tree Source Codes for Optimal Compressed Tree Data Structures and Range Minima},
	author = {Munro, J. Ian and Nicholson, Patrick K. and Seelbach Benkner, Louisa and Wild, Sebastian},
	booktitle = {European Symposium on Algorithms ({ESA})},
	doi = {10.4230/LIPICS.ESA.2021.70},
	publisher = {Schloss Dagstuhl - Leibniz-Zentrum f\"{u}r Informatik},
	year = {2021},
	editor = {Mutzel, P. and Pagh, R. and Herman, G},
	pages = {70:1--70:18},
	archivePrefix = {arXiv},
	eprint = {2104.13457},
	arxivId = {2104.13457},
	primaryClass = {cs.DS}
}

@book{johnson2005univariate,
  title={Univariate discrete distributions},
  author={Johnson, Norman L and Kemp, Adrienne W and Kotz, Samuel},
  year={2005},
  publisher={John Wiley \& Sons}
}

@inproceedings{Patrascu2008,
	author = {P{\u a}tra{\textcommabelow s}cu, Mihai},
	Booktitle = {Symposium on Foundations of Computer Science (FOCS)},
	Doi = {10.1109/focs.2008.83},
	Month = oct,
	Publisher = {{IEEE}},
	Title = {Succincter},
	Year = {2008}}

@article{GearyRamanRaman2006,
  title = {Succinct ordinal trees with level-ancestor queries},
  doi = {10.1145/1198513.1198516},
  year  = {2006},
  month = oct,
  publisher = {{ACM}},
  volume = {2},
  number = {4},
  pages = {510--534},
  author = {Richard F. Geary and Rajeev Raman and Venkatesh Raman},
  journal = {{ACM} Transactions on Algorithms}
}

@inproceedings{HeMunroNekrichWildWu2020,
	title = {Distance Oracles for Interval Graphs via Breadth-First Rank/Select in Succinct Trees},
	author = {Meng He and J. Ian Munro and Yakov Nekrich and Sebastian Wild and Kaiyu Wu},
	booktitle = {International Symposium on Algorithms and Computation (ISAAC)},
	pages = {25:1--25:18},
	series = {LIPIcs},
	publisher = {Schloss Dagstuhl},
	archivePrefix = {arXiv},
	eprint = {2005.07644},
	arxivId = {2005.07644},
	primaryClass = {cs.DS},
	year = {2020}}

@article{EdmondsPolytope,
  author       = {Jack Edmonds},
  title        = {Matroids and the greedy algorithm},
  journal      = {Math. Program.},
  volume       = {1},
  number       = {1},
  pages        = {127--136},
  year         = {1971},
  url          = {https://doi.org/10.1007/BF01584082},
  doi          = {10.1007/BF01584082},
  bibsource    = {dblp computer science bibliography, https://dblp.org}
}

@inproceedings{almost3colouring,
author = {Dinur, Irit and Khot, Subhash and Perkins, Will and Safra, Muli},
title = {Hardness of Finding Independent Sets in Almost 3-Colorable Graphs},
year = {2010},
isbn = {9780769542447},
publisher = {IEEE Computer Society},
address = {USA},
doi = {10.1109/FOCS.2010.84},
booktitle = {Proceedings of the 2010 IEEE 51st Annual Symposium on Foundations of Computer Science},
pages = {212–221},
numpages = {10},
series = {FOCS '10}
}

\clearpage
\appendix
\ifkoma{\addpart{Appendix}}{}

\section{Hardness of Hitting Set}\label{app:gt05}

In this section, we walk through the proof of~\cite{GuruswamiTrevisan2005} for proving the hardness of $(\alpha, r, k)$-APX-EHS for every $\alpha < e$ and for some fixed $r, k$ depending only on $\alpha$. We first give an equivalent formulation of this problem as a CSP.

\begin{definition}
    In $(\alpha, r, k)$-APX-EHS, we are given an $r$-regular $k$-uniform hypergraph $H$. A solution is a set $S$ of vertices. The value of a solution is the proportion of edges $T$ such that $|S \cap T| = 1$. We must decide whether there exists a solution with value 1, or not even a solution with value $1 / \alpha$.
\end{definition}

To reiterate,~\cite{GuruswamiTrevisan2005} proves that for every $\alpha < e$ there exists $k$ such that this problem is \NP-hard \emph{but without the $r$-regularity condition}. However, their reduction implicitly produces an $r$-regular instance, where $r$ depends only on $\alpha$. In this section we will merely show this fact, going through the reduction, without reproving other facts.

The reduction of~\cite{GuruswamiTrevisan2005} starts from the $p$-prover proof system of Feige~\cite{feige98setcover}, used to prove hardness results for set cover. Let us first define this proof system.

\begin{definition}\label{def:feige}
    Fix constants $p \geq 2$ and $u$, an even integer. An instance of size $n$ of the $p$-prover system is defined as follows. Let $R = (5n)^u$, $Q = n^{u / 2} (5n/3)^{u / 2}$, $A = 2^u$, $B = 4^u$. The instance is a hypergraph $H = (V, E)$, where each edge is equipped with a sequence of functions from $[B]$ to $[A]$, with the following properties.
    \begin{description}
    \item[Uniformity.] The hypergraph is $p$-uniform and $p$-partite, i.e.~the vertex set is given by $V = Q_1 \cup \ldots \cup Q_p$, and each edge takes exactly one vertex from each part $Q_i$.
    \item[Regularity.] The hypergraph is regular with regularity
    \[
    \frac{R}{Q}
    \wwrel=
    \frac{(5n)^u}{n^{u/2} (5n/3)^{u/2}}
    \wwrel=
    15^{u/2}.
    \]
    \item[Vertices.] Each part $Q_i$ has size $|Q_i| = Q$.
    \item[Edges.] There are $R$ edges in the hypergraph. Edge $r \in [R]$ is denoted by $(q_{r1}, \ldots, q_{rp})$, where $q_{ri} \in Q_i$.
    \item[Projections.] Each edge $(q_{r1}, \ldots, q_{rp})$ is equipped with a sequence of projections $(\pi_{r1}, \ldots, \pi_{rp})$, where each $\pi_{ri}$ is a function from $[B]$ to $[A]$ that is $(B/A) = 2^u$-to-1, i.e.~the preimage through $\pi_{ri}$ of every value has size $2^u$.
    \end{description}

    A solution to this problem is a labelling $a : V \to [B]$. Such a solution:
    \begin{itemize}
        \item \emph{Strongly satisfies} edge $r$ if $\pi_{ri}(a(q_{ri})) = \pi_{rj}(a(q_{rj}))$ for \emph{all} $i, j \in [p]$.
        
        \item \emph{Weakly satisfies} edge $r$ if $\pi_{ri}(a(q_{ri})) = \pi_{rj}(a(q_{rj}))$ for \emph{at least one} $i, j \in [p]$ with $i \neq j$.
    \end{itemize}
\end{definition}

Feige then proves the following key \NP-hardness result.
\begin{theorem}
    For some universal constant $0 < c < 1$, for every $p \geq 2$ and for $u$ large enough, the following problem is \NP-hard. Given an instance of the $p$-prover system with parameter $u$, it is \NP-hard to distinguish whether:
    \begin{description}
        \item[Yes instance.] There exists a labelling that strongly satisfies all edges.
        \item[No instance.] No labelling weakly satisfies even a $p^2 c^u$  fraction of the edges.
    \end{description}
\end{theorem}

To make this notion more intuitive, we explain why it is called a \emph{proof system}. Imagine that each part $i \in [p]$ corresponds to one of $p$ provers, which are trying to prove some fact to a checker. Each node $q_i \in Q_i$ is a query that the checker could ask prover $i$. A solution is a strategy for answering these queries. The checker selects an edge of the hypergraph at random, and asks the corresponding queries to each prover. The checker then passes these answers through the projection functions for the given edge. They are then strongly convinced of the correctness of the proof if the answer is unanimous; and weakly convinced if at least two of the answers are equal. 

Feige's hardness result essentially says the following: for any instance of e.g.~3-SAT, it is possible, in polynomial time, to create a $p$-prover system where (i) if the 3-SAT instance is satisfiable then there exists a proof that strongly convinces the checker with probability~1, and (ii) if the 3-SAT instance is unsatisfiable, then no proof even weakly convinces the checker with probability $p^2 c^u$.

We now give the reduction of~\cite{GuruswamiTrevisan2005}, where we trace the $r$-regularity, as well.

\begin{theorem}
    Fix $\epsilon > 0$, and let $1 / \alpha = 1 / e + \epsilon$. There exists a positive integer $p$ such that for all large enough $u$, the following holds. Let $k = 4^u$ and $r = 15^u p^{2^u - 1}$. Given an instance of the $p$-prover proof system with parameter $u$, we can create in polynomial time an instance of $(\alpha, r, k)$-APX-EHS that is a Yes resp.\ No instance whenever the $p$-prover proof system we were given was a Yes resp.\ No instance. Hence, $(\alpha, r, k)$-APX-EHS is \NP-hard.
\end{theorem}
\begin{proof}
    Suppose that the input $p$-prover system with parameter $u$ is denoted by the same symbols as in \wref{def:feige}. We must now output an $r$-regular $k$-uniform hypergraph with the required properties. The hypergraph we output is constructed as follows.
    \begin{description}
        \item[Vertices.] For every $i \in [p]$, $q \in Q_i$ and $b \in [B]$, we introduce $(i, q, b)$ as a vertex.
        \item[Edges.] For every $r \in [R]$ and $\mathbf{x} = (x_1, \ldots, x_A) \in [p]^A$, we introduce an edge
        \[
        S_{r\mathbf{x}} \wrel= \{ (i, q_{ri}, b) \mid x_{\pi_{ri}(b)} = i \}.
        \]
    \end{description}
    The soundness and completeness of this reduction is proved by~\cite{GuruswamiTrevisan2005}~-- hence if the original instance was a yes/no-instance of the $p$-prover game with parameter $u$, this hypergraph is a yes/no-instance of $(\alpha, r, k)$-APX-EHS. What remains is to show $k$-uniformity and $r$-regularity; the first is also shown by~\cite{GuruswamiTrevisan2005}, but we give the proof again for completeness.
    \begin{description}
        \item[Uniformity.] We must count the size of each set $S_{r\mathbf{x}}$. Let us fix the \emph{output} of $\pi_{ri}(b)$ in the definition of this set. Then:
        \[
        S_{r\mathbf{x}} \wrel= \bigcup_{a \in [A]} \{ (i, q_{ri}, b) \mid i = x_a, \pi_{ri}(b) = a \}.
        \]
        Now note that for every fixed $a$, we may select $b$ in $B / A = 2^u$ ways in order to make $\pi_{ri}(b) = a$, since $\pi_{ri}$ is a $2^u$-to-1 function. Then, $i$ is fixed as a function of $a$ and $\mathbf{x}$, and likewise $q_{ri}$ is fixed as a function of $r$ and $i$. So, we see that for each choice of $a \in [A]$, $S_{r\mathbf{x}}$ contains precisely $2^u$ different values. It follows that $|S_{r\mathbf{x}}| = |A| 2^u = 4^u = k$.
        \item[Regularity.] We must count how many times any fixed vertex $(i, q, b)$, where $q \in Q_i$ and $b \in [B]$,  belongs to an edge $S_{r \mathbf{x}}$. For this to be the case, we first need that $q = q_{ri}$. This holds for exactly $15^{u / 2}$ choices of $r$, since the original $p$-prover system instance is $15^{u / 2}$-regular. Now, fix one such choice of $r$ and assume $q = q_{ri}$. Let us count for how many vectors $\mathbf{x} = (x_1, \ldots, x_A) \in [p]^A$ we have $(i, q_{ri}, b) \in S_{r\mathbf{x}}$. The only condition for this to be the case is that $x_{\pi_{ri}(b)} = i$. Since we have fixed $(i, q_{ri}, b)$ and $r$, the index $\pi_{ri}(b)$ is fixed. Hence there are exactly $p^{A - 1}$ choices of $\mathbf{x}$ that lead to $(i, q_{ri}, b) \in S_{r \mathbf{x}}$. Overall, we see that the instance is regular with regularity
        \[
            15^u p^{A - 1} \wrel= 15^u p^{2^u - 1} \wrel= r.
        \]
    \end{description}
    This concludes the proof.
\end{proof}

\section{A Tree is Never Worse Than a Forest}
\label{app:tree-not-forest}

One may wonder whether a \emph{strictly} lower indegree entropy can be achieved by removing a spanning non-tree forest instead of a tree. This section answers this question with a \emph{No}; there is always a spanning tree whose elimination from the graph produces an optimal indegree entropy.
Indeed, we show that removing an additional edge never increases $H(G)$, and strictly decreases $H(G)$ (unless $H(G)=0$ already).

First of all, let us prove a simple arithmetic claim.

\begin{lemma}
    \label{lem:logs-slope}
    For $a,b \in \mathbb{R}_{\geq 1}$ and $a \geq b \geq 1$, we have 
    \[(a-1)\lg(a-1)-(b-1)\lg(b-1) \wwrel\leq a\lg a-b\lg b,\]
    with equality only for $a=b$.
\end{lemma}
\begin{proof}
    After dividing by $\ln 2$ and rearranging, it suffices to show that 
    \[f(x) \wwrel= (x-1)\ln(x-1)-x\ln x
    \] 
    is a strictly decreasing function for $x\in\R_{>1}$. The derivative $f'(x)$ is computed as:
    \[
    		f'(x)
    	\wwrel=
    		\frac{\mathrm{d}}{\mathrm{d}x}(x-1)\ln(x-1)-x\lg x
    	\wwrel=
    		\ln(x-1)-\ln(x)
    	\wwrel=
    		\ln(1-1/x),
    \]
    which is $f'(x)<0$ for any $x>1$.
\end{proof}

Now, we are ready to show the main claim of this section.
Let $\mathcal{F}(G)$ and $\mathcal{T}(G) \subseteq \mathcal{F}(G)$ denote the set of spanning forests and spanning trees of $G$, respectively. 

\begin{lemma}\label{lem:spanningTreeOptimal}
     Let $G=(V,E)$ be a weakly connected, directed graph. Then there exists $T \in \mathcal{T}(G)$ such that \[H(G-T)\wwrel=\min_{F \in \mathcal{F}(G)}H(G-F).\]
\end{lemma}
\begin{proof}
    We actually show the stronger claim hinted at above: 
    apart from the degenerate case of a star graph with indegree entropy $0$, 
    removing any edge strictly improves the resulting indegree entropy.
    Formally, recall that $H(G) = -\sum_{v \in V}d_v\lg \bigl(\frac{d_v}{m}\bigr)$. 
    We show that $H(G-e) \leq H(G)$ for any $e \in E$, with equality only for $H(G) = 0$.
    This immediately proves the claim since spanning trees are precisely the maximal spanning forests of $G$.
	
    We rewrite the indegree entropy for $G$ as
    \begin{align*}
        	H(G)
        &\wwrel= 
        	-\sum_{v \in V}d_v\lg \left(\frac{d_v}{m}\right) 
    	\wwrel= 
    		-\sum_{v \in V}d_v\lg d_v + m\lg m.
    \end{align*}
    Removing $e=(u,w)$ from $G$ makes exactly one indegree smaller by $1$: 
   	$d_{w}$ decreases upon deleting $e$ from $G$. 
   	We have
    \begin{align*}
        	H(G-e) 
        &\wwrel=
        	-\;\sum_{\mathclap{v \in V \setminus \{w\}}}\; d_v\lg d_v 
        	\bin- (d_{w}-1) \lg (d_{w}-1)
        	\bin+(m-1)\lg (m-1),
    \intertext{where $d_{w}$ is the indegree of $w$ in $G$ (before deleting $e$). Comparing with}
        	H(G)
        &\wwrel= 
        	-\;\sum_{\mathclap{v \in V \setminus \{w\}}} \; d_v\lg d_v  
        	\bin- d_{w} \lg d_{w}
        	\bin+m\lg m
    \end{align*}
    yields that $H(G-e) < H(G)$ if and only if 
    \begin{equation}
    \label{eq:HGme-vs-HG}
    	(m-1)\lg (m-1) - (d_{w}-1) \lg (d_{w}-1)
    	\wwrel<
    	m\lg m - d_{w} \lg d_{w}.
    \end{equation}
    Since $1 \leq d_{w} \leq m$, we can apply \wref{lem:logs-slope} with $a=m$ and $b=d_w$
    to obtain \wref{eq:HGme-vs-HG}, with equality only for $m=d_w$,
    \ie{} when in $G$ all edges point to $w$.
\end{proof}

\section{Copy Model}
\label{app:copy-model}

In this section, we describe the copy model for generating random graphs studied in~\cite{TurowskiMagnerSzpankowski2020}. 
It is loosely motivated by protein interaction networks, in which mutations of the gene coding for a known protein creates a new, but functionally equivalent protein, which then shares the interaction relations of its ancestor.
In the random copy model of~\cite{TurowskiMagnerSzpankowski2020}, we start with a seed graph $G_0$ and repeatedly copy a vertex chosen uniformly at random among the current graph, thus creating a twin of the copied vertex.
Turowski et al.\ compute the entropy of the resulting graph distribution, both for the labelled and unlabelled case, and describe corresponding compression methods, but without query support.
We show that our TREX data structure with twin removal 
yields on asymptotically \emph{instance-optimal} representation of unlabelled graphs from the copy model 
that includes full support of efficient navigational operations.

\subsection{The Random Copy Model}

We recall the definitions and notation from~\cite{TurowskiMagnerSzpankowski2020}. 
(We match this to our notation from \wref{sec:twin-removal} below.)

The random copy model is parameterized by an initial graph $G_0$ on $n_0$ vertices. For each step $1 \leq i \leq n$, the graph $G_i$ is constructed from $G_{i-1}$ as follows:

\begin{enumerate}
\item Select a vertex $v \in V(G_{i-1})$ uniformly at random.
\item Add a new vertex $v_i$ to the graph.
\item Connect $v_i$ to all neighbors of $v$.
\end{enumerate}

Note that $v$ and $v_i$ themselves are not connected.

After $n$ steps, the graph $G_n$ has $n_0 + n$ vertices. The structure of $G_n$ depends strongly on the choice of $G_0$. 
The original vertices of $G_0$ are denoted $U = V(G_0) = \{ u_1, \dots, u_{n_0}\}$,
and the vertices added in the duplication process are denoted $V = \{ v_1, \dots, v_n \}$.
Moreover, we denote by $N_n(v)$ the neighbourhood of vertex $v$ in $G_n$, that is, all vertices that are
adjacent to $v$ in $G_n$.

Note that Turowski et al.\ considered undirected graphs, 
but all the above definitions can be easily generalised to directed graphs.

\subsection{Properties}
\label{sec:properties}

We now define the concepts of parent and ancestor of a vertex. 
A vertex $w$ is called the \emph{parent} of $v$
if $v$ was copied from $w$ at some step $1 \leq i \leq n$.
By repeatedly moving from $v$ to its parent, we eventually reach a vertex $u\in V(G_0)$; this vertex $u$ is called the \emph{ancestor} of $v$. 
Conversely, we define the set of \emph{descendants} of a vertex $u_i \in V(G_0)$ in $G_n$
as the set of all vertices $v\in V(G_n)$ whose ancestor is $u_i$.
For $i=1,\ldots,n_0$, we denote the descendants of $u_i$ by $\mathcal C_{i,n}$.

The key observation about the copy model is the invariant that all descendants of a vertex $u_i$ are \emph{twins}.
This is initially true when a new twin is added to the graph, and later additions of twins cannot separate a vertex $v$ and its ancestor $u$ since cloning any neighbour of $v$ resp.\ $u$ would also include $u$ resp.\ $v$; (see also Lemma 1 of Turowski et al.~\cite{TurowskiMagnerSzpankowski2020}).

Let~$C_{i,n} := |\mathcal C_{i,n}|$, that is, the number of vertices in $G_n$ that are the twins of $u_i$ (including $u_i$ itself).
It is not hard to see that the sequence of variables $(C_{i,n})_{i=1}^{n_0}$ can be described as an urn model (indeed the classical Pólya urn~\cite{Mahmoud2008,johnson2005univariate}): The urn contains balls of $n_0$ different ``colors'' $u_1,\ldots,u_{n_0}$. At time $n=0$, we have exactly one ball per color for a total of $n_0$ balls. 
In each step, we pick a ball uniformly at random from the urn and return both the chosen ball, as well as one \emph{additional} ball of the same color to the urn.
The joint distribution of ``new'' color multiplicities $(C_{i,n}-1)_{i=1}^{n_0}$ is then Dirichlet-multinomial $\mathrm{DirMult}(n; \alpha_1, \ldots, \alpha_{n_0})$ distributed with parameters $\alpha_i=1$ for $1 \le i \le n_0$ (the initial multiplicities of colors in the urn).
$\mathrm{DirMult}(n;1,\ldots,1)$ is a \emph{uniform} distribution over all vectors that sum to $n$:
\begin{equation}
\label{eq:prob-copy-model}
\Prob[\big]{(C_{i,n})_{i=1}^{n_0} = (k_i+1)_{i=1}^{n_0}} \wwrel=
\begin{dcases}
\dfrac1{\binom{n+n_0-1}{n}}, & \text{if } \sum_{i=1}^{n_0} k_i = n, \\
0, & \text{otherwise}.
\end{dcases}
\quad \forall {1 \leq i \leq n_0}: k_i \in \mathbb{N}_{\ge1}
\end{equation}

\subsection{Entropy}

Under the assumption that the initial graph $G_0$ is fixed/known and \emph{asymmetric}, \ie\ its automorphism group is of size 1,
Turowski et~al.~\cite{TurowskiMagnerSzpankowski2020} determined the asymptotic entropy for both unlabeled and labelled random graphs generated by this model.
They proved~\cite[Theorem 2]{TurowskiMagnerSzpankowski2020} that the entropy of \emph{labeled} graphs is 
\begin{equation}
	(H_{n_0}-1)\lg(e) \cdot n \bin+ \frac{n_0-1}{2}\cdot \lg n \bin\pm O(n_0 \log n_0),
	\qquad \text{with } H_k = 1+\frac12+ \frac13+ \cdots +\frac1k
\end{equation}
while the entropy of the unlabelled graph (the entropy of the isomorphism class of a random graph from
the model, \aka the ``structural entropy'') is significantly smaller~\cite[Theorem~1]{TurowskiMagnerSzpankowski2020}: 
\[
	(n_0-1) \lg n \bin\pm O(n_0 \log n_0).
\]
Note that for constant $n\to\infty$, the labelled graph entropy grows exponentially faster than the unlabelled one ($\Theta(n)$ vs. $\Theta(\log n)$);
the copy model is hence an extreme case where the vast majority of information in a (labelled) graphs 
lies in the vertex labels, not the underlying graph structure.

In~\cite[Theorems 4, 5]{TurowskiMagnerSzpankowski2020}, they also provide optimal algorithms for compressing both unlabeled and labelled graphs generated by the full copy model. It is based on arithmetic coding and is optimal up to two bits. This is provided, we have already compressed~$G_0$. This matches the entropies up to a constant additive term.
Neither representation is known to support efficient queries.

\subsection{TREX with twin removal is instance-optimal}

Beyond the entropy of the distribution, from \wref{eq:prob-copy-model} it is indeed easy to compute the instance-specific information content $\lg(1/\Prob{G \given G_0})$ of a graph $G$ drawn according to the copy model with seed graph $G_0$ (again assuming $G_0$ is asymmetric).
Given $G_0$ and $n$, the probability of a graph $G$ to arise in the copy model is 
$\Prob{G\given G_0, n} = 1\big/\binom{n+n_0-1}{n}$, provided the twin reduction of $G$ is $G_0$ (otherwise, $\Prob{G\given G_0, n} = 0$).
In the former case, we have
\begin{align}
		\lg\bigl(1/\Prob{G\given G_0,n}\bigr)
	&\wwrel=
		\lg \binom{n+n_0-1}{n_0-1}
		\notag
\\	&\wwrel=
		(n_0-1)\lg\left(\frac{n+n_0-1}{n_0-1}\right) \wbin+ O(n_0)
		\notag
\\	&\wwrel=
		(n_0-1)\lg(n+n_0-1) \wbin\pm O(n_0 \log n_0)
	\label{eq:Prob-G-given-G0}
\end{align}

Note that in our notation from \wref{sec:twin-removal}, the overall graph size is $N = n+n_0$ and the \emph{seed graph} $G_0$ in the copy model is the \emph{twin reduction} of $G$ in twin removal. 
In \wref{sec:twin-removal}, we denoted the size of the twin reduction by $n=|V(G_0)|$ instead of $n_0$ here. 
The \emph{representative} of the twin class in \wref{sec:twin-removal} is called the \emph{ancestor} of a vertex in $G$ here.

With this, we find that the claimed space in \wref{thm:space-time-analysis-G}
\emph{almost} matches $\lg\bigl(1/\Prob{G\given G_0,n}\bigr)$~-- but not quite!
However, a closer look at bitvector $B$ reveals that its first bit is \emph{always} $1$, and we thus need not encode it at all.  
The bits we actually have to encode are only $N-1$ bits, of which exactly $n-1$ are $1$. 
The space for $B$ thus becomes $\lg\binom {N-1}{n-1} \le (n-1)\lg(N-1) + o(N)$, which is asymptotically for large $N$ the instance-optimal space usage to represent the unlabeled graph $G$ given $G_0$ (which we encode using TREX).

\end{document}